\documentclass{article}

\usepackage{microtype}
\usepackage{graphicx}
\usepackage{booktabs} 

\usepackage{hyperref}

\usepackage[accepted]{icml2024}

\usepackage{amsmath}
\usepackage{amssymb}
\usepackage{mathtools}
\usepackage{amsthm}

\usepackage[capitalize,noabbrev]{cleveref}

\theoremstyle{plain}
\newtheorem{theorem}{Theorem}[section]

\theoremstyle{definition}

\theoremstyle{remark}

\usepackage[textsize=tiny]{todonotes}
\usepackage{xcolor,import}
\usepackage{listings}
\usepackage{subfigure}
\usepackage[export]{adjustbox}
\usepackage[super]{nth}
\usepackage{mathtools}

\graphicspath{{figures/}}

\newcommand{\website}{https://sites.google.com/view/dico-marl}
\newcommand{\benchmarl}{https://github.com/facebookresearch/BenchMARL}
\newcommand{\vmas}{https://github.com/proroklab/VectorizedMultiAgentSimulator}

\newcommand{\R}{\mathbb{R}}
\icmltitlerunning{Controlling Behavioral Diversity in Multi-Agent Reinforcement Learning}

\begin{document}

\twocolumn[
\icmltitle{Controlling Behavioral Diversity in Multi-Agent Reinforcement Learning}



\icmlsetsymbol{equal}{*}

\begin{icmlauthorlist}
\icmlauthor{Matteo Bettini}{cam}
\icmlauthor{Ryan Kortvelesy}{cam}
\icmlauthor{Amanda Prorok}{cam}

\end{icmlauthorlist}

\icmlaffiliation{cam}{Department of Computer Science, University of Cambridge, Cambridge, UK}

\icmlcorrespondingauthor{Matteo Bettini}{mb2389@cl.cam.ac.uk}

\icmlkeywords{Multi-Agent Reinforcement Learning, Diversity, Control}

\vskip 0.3in
]



\printAffiliationsAndNotice{}  

\begin{abstract}
The study of behavioral diversity in Multi-Agent Reinforcement Learning (MARL) is a nascent yet promising field.
In this context, the present work deals with the question of how to \emph{control} the diversity of a multi-agent system. 
With no existing approaches to control diversity to a set value, current solutions focus on blindly promoting it via intrinsic rewards or additional loss functions, effectively changing the learning objective and lacking a principled measure for it.
To address this, we introduce Diversity Control (DiCo), a method able to control diversity to an exact value of a given metric by representing policies as the sum of a parameter-shared component and dynamically scaled per-agent components.  
By applying constraints directly to the policy architecture, DiCo leaves the learning objective unchanged, enabling its applicability to any actor-critic MARL algorithm.
We theoretically prove that DiCo achieves the desired diversity, and we provide several experiments, both in cooperative and competitive tasks, that show how DiCo can be employed as a novel paradigm to increase performance and sample efficiency in MARL.
Multimedia results are available on the paper's \href{\website}{website}\footnote{\url{\website}\label{link:experiments}}.
\end{abstract}

\section{Introduction}

Diversity is key to collective intelligence~\cite{woolley2015collective} and commonplace in natural systems~\cite{kellert1997value}. 
Just as biologists and ecologists have demonstrated the role of functional diversity in ecosystem survival~\cite{cadotte2011beyond}, it has also been shown to provide important performance benefits in Multi-Agent Reinforcement Learning (MARL)~\cite{bettini2023hetgppo, chenghao2021celebrating}.
Despite this, methods that leverage diversity in MARL are understudied, with the field still in its infancy. In particular, one question that remains unanswered is how to control a system's diversity to an exact, quantified value. This is the focus of this work.

Behavioral diversity in MARL is intrinsically tied to the concept of policy parameter sharing~\cite{christianos2021scaling}. 
When agents share policy parameters, they obtain higher sample efficiency, but learn a single homogeneous policy. When agents do not share parameters, they are able to learn heterogeneous policies, but achieve lower sample efficiency.
Following the latter paradigm, several methods have been proposed to promote behavioral diversity, showing its utility in many MARL problems~\cite{jiang2021emergence, chenghao2021celebrating,mahajan2019maven,wang2020roma}. 
These approaches aim to boost diversity by employing information theoretical objectives, such as additional intrinsic rewards or secondary losses.
In doing so, they effectively change the learning objective, without being able to measure the resulting diversity.
In contrast, rather than than boosting diversity, we propose a novel problem formulation that aims to \textit{control diversity to the exact value of a given metric}, leaving the learning objective unchanged.
This paradigm opens new avenues for the control and study of diversity in MARL, presenting a novel tool that researchers can use to constrain agents to different diversity levels, aiding in the discovery of emergent strategies that leverage the benefits of diversity.

We propose Diversity Control (DiCo), the first method to control behavioral diversity in MARL to a desired value, applicable to any actor-critic algorithm with stochastic or deterministic continuous actions.
Bypassing issues (mentioned above) that change the structure of the learning objective, DiCo finds a new approach to diversity control.
It represents policies as the sum of a parameter-shared (homogeneous) component and individual (heterogeneous) components, which are dynamically scaled according the the current and desired value of a given diversity metric.
We provide theoretical proofs that DiCo achieves the desired diversity and demonstrate it empirically in a didactic case-study (i.e., \textit{Multi-Agent Navigation}).
Whereas previous works lacked tools to inspect the learned policies, we present novel visualizations of the policies' diversity distribution, which enable to analyze how the agent policies learn to distribute the diversity dictated by DiCo over the observation space.

To showcase some possible applications of DiCo, we run experiments in four cooperative and competitive MARL tasks, comparing DiCo agents (constrained to different levels of diversity) to homogeneous and unconstrained heterogeneous agents.
Our experiments show that, by constraining the policy search space, DiCo can achieve higher performance, exploration, and sample efficiency.
Via trajectory analyses of the learned DiCo policies, we show how different diversity constraints can lead to the discovery of multiple novel emergent strategies in various tasks.
Multimedia experiment results are available on the paper's \href{\website}{website}\footref{link:experiments}.

Our contributions are as follows:
\begin{itemize}
\setlength\itemsep{0.05pt}
    \item DiCo, the first method to control behavioral diversity in MARL to a set value of a given diversity metric, complemented with theoretical proofs and empirical demonstrations of achieving the desired diversity;
    \item Novel visualization tools to investigate the policies' diversity distribution, that showcase how agents distribute the diversity dictated by DiCo over the observation space;
    \item A set of MARL experiments in cooperative and competitive tasks, demonstrating several benefits of the DiCo method and the emergence of novel strategies. In particular, we gather the following insights:
    \begin{itemize}
        \item In tasks that benefit from diversity, where unconstrained heterogeneous policies learn a suboptimal low-diversity strategy, we show how DiCo can be employed to find better-performing and more diverse strategies faster;
        \item In tasks that benefit from homogeneity, where unconstrained heterogeneous policies prove less sample efficient and too diverse, we show how DiCo can find less-diverse optimal policies faster;
        \item We show how DiCo, by constraining the search space of heterogeneous policies, can be used to find novel emergent strategies.
    \end{itemize}
\end{itemize}

\section{Related Works}
Diversity in MARL has recently gained increasing attention.
In particular, several works have outlined the caveats and implications of parameter sharing across agent policies~\cite{christianos2021scaling,fu2022revisiting,bettini2023hetgppo}. 
These works show that, despite homogeneous policies could emulate diverse behavior based on the input context, parameter sharing can impede the learning of diverse behavioral roles, prove brittle with respect to noise, or completely prevent success in certain tasks. 
In the following paragraphs, we review existing approaches to promote behavioral diversity in MARL. Additional related works, including diversity in population-based RL, are presented in \autoref{app:related}.

\textbf{Diversity via Intrinsic Reward.}
A common solution to promote diversity among agents in a system is to design an intrinsic reward that is added to the task reward, creating an auxiliary objective for the agents.
\citet{jaques2019social} introduce an intrinsic reward based on mutual information between the agents' actions in order to promote social influence.
\citet{wang2019influence} propose another influence-based reward structure to promote exploration, additionally considering the influence on other agents' transition and reward functions.
\citet{jiang2021emergence} introduce an intrinsic reward based on the mutual information between an agent’s identity and observation.
Lastly, \citet{chenghao2021celebrating} consider agents equipped with both shared and individual networks, that are optimized using an intrinsic reward based on the mutual information between an agent’s identity and its trajectory.
While intrinsic motivation via additional rewards constitutes a powerful diversity boosting mechanism, it could introduce learning issues due to its interaction with the primary task reward and, thus, require extensive re-tuning when applied to new tasks. Furthermore, it modifies the RL objective, which can impact the optimality of the learned policy in the original problem. For these reasons, we do not consider intrinsic reward structures in our diversity control paradigm.

\textbf{Diversity via Objective Function.}
Another solution to promote diversity considers designing an auxiliary loss as a diversity booster.
In MAVEN \cite{mahajan2019maven} agents condition their behavior on a shared latent variable controlled by a hierarchical policy. MAVEN's objective then maximizes mutual information between the trajectories and latent variables to learn a diverse set of behaviors.
ROMA \cite{wang2020roma} introduces a diversity regularizer to incentivize learning specialized roles in cooperative tasks.
The addition of a secondary objective, however, can impact the main learning objective, requiring careful task-dependent tuning. In contrast, our approach does not employ additional objective functions to control diversity.

All the approaches presented in this section devise additional objectives to promote diversity across agents in a system, either in the form of an intrinsic reward term or an auxiliary loss. Furthermore, most works focus on discrete action spaces, where it is not possible to quantify action diversity using a continuous metric.
In contrast, our work enables controlling diversity in continuous action spaces to the exact desired value of a given diversity metric by architecturally constraining the agents' policy networks, thus obviating the need for additional diversity-based objectives or rewards.

\section{Background}
In this section, we introduce the task formulation and the behavioral diversity metric that will be utilized in our problem formulation and method.

\textbf{Partially Observable Markov Games.} A Partially Observable Markov Game (POMG)~\cite{shapley1953stochastic} is defined as a tuple
$$\left \langle \mathcal{N}, \mathcal{S}, \left \{ \mathcal{O}_i \right \}_{i \in \mathcal{N}}, \left \{ \sigma_i \right \}_{i \in \mathcal{N}},  \left \{ \mathcal{A}_i \right \}_{i \in \mathcal{N}}, \left \{ \mathcal{R}_i \right \}_{i \in \mathcal{N}}, \mathcal{T}, \gamma \right \rangle,$$
where $\mathcal{N} = \{1,\ldots, n\}$ denotes the set of agents,
$\mathcal{S}$ is the state space, and,
$\left \{ \mathcal{O}_i \right \}_{i \in \mathcal{N}}$ and
$\left \{ \mathcal{A}_i \right \}_{i \in \mathcal{N}}$
are the observation and action spaces, with $\mathcal{O}_i \subseteq \mathcal{S}, \; \forall i \in \mathcal{N}$. 
Further, $\left \{ \sigma_i \right \}_{i \in \mathcal{N}}$ 
and
$\left \{ \mathcal{R}_i \right \}_{i \in \mathcal{N}}$
are the agent observation and reward functions (potentially identical for all agents\footnote{In this work we consider also problems represented as Dec-POMDPs~\cite{bernstein2002complexity}, which are a particular subclass of POMGs with one shared reward function.}), such that
$\sigma_i : \mathcal{S} \mapsto \mathcal{O}_i$, and,
$\mathcal{R}_i: \mathcal{S} \times \left \{ \mathcal{A}_i \right \}_{i \in \mathcal{N}} \times \mathcal{S} \mapsto \R$.
$\mathcal{T}$ is the stochastic state transition model, defined as $\mathcal{T} : \mathcal{S} \times \left \{ \mathcal{A}_i \right \}_{i \in \mathcal{N}}   \mapsto  \Delta\mathcal{S}$, which outputs the probability $\mathcal{T}(s^t, \left \{ a^t_i \right \}_{i \in \mathcal{N}},s^{t+1})$ of transitioning to state $s^{t+1} \in \mathcal{S}$ given the current state $s^t \in \mathcal{S}$ and actions $\left \{ a^t_i \right \}_{i \in \mathcal{N}}$, with $a^t_i \in \mathcal{A}_i$. $\gamma$ is the discount factor.
Agents are equipped with (possibly stochastic) policies $\pi_{i}(a_i|o_i)$, which compute an action given a local observation. 

\textbf{System Neural Diversity.}
Several measures have been proposed to quantify behavioral diversity among agent policies~\cite{mckee2022quantifying,yu2021informative,hu2022policy,bettini2023snd}.
In this work, we employ System Neural Diversity (SND)~\cite{bettini2023snd} as it is the only metric that can be computed in closed-form between the continuous action distributions provided as output by the policies. While the theorem presented in this work leverages the properties of this metric, other metrics could equivalently be used as long as the main theorem results are proven to hold. 
In \autoref{app:general_metric} we: (1) present a theorem that defines the properties that a diversity metric needs to satisfy in order to be used in DiCo, and (2) provide further motivation for the use of SND in this work.
SND is a behavioral diversity metric that measures heterogeneity in a multi-agent system. It is computed in two phases. First, the Wasserstein statistical metric~\cite{vaserstein1969markov} is used to measure the diversity between two agents' policies over a set $O$ of observations (generated via rollouts): $d(\pi_i,\pi_j) = \frac{1}{|O|}\sum_{o\in O }W_2(\pi_i(o),\pi_j(o))$. This distance can be computed in closed-form for deterministic policies and for stochastic policies outputting multivariate normal distributions. Pairwise behavioral distances are then aggregated in a system-level metric by taking the mean over agent pairs 
$\mathrm{SND}(\left \{ \pi_i \right \}_{i \in \mathcal{N}}) = \frac{2\sum_{i=1}^n\sum_{j=i+1}^n d(\pi_i,\pi_j)}{n(n-1)}$, leading to the following expression for $\mathrm{SND}$:
\begin{multline}
    \mathrm{SND}(\left \{ \pi_i \right \}_{i \in \mathcal{N}}) =\\ \frac{2}{n(n-1)|O|}\sum_{i=1}^n\sum_{j=i+1}^n \sum_{o\in O }W_2(\pi_i(o),\pi_j(o)).
    \label{eq:snd}
\end{multline}

The choice of $\mathrm{SND}$ over alternative measures is motivated by several desirable properties. Firstly, by using $W_2$, SND presents all the properties of a statistical metric~\cite{menger2003statistical}, which would not hold for diversity measures based on statistical divergences (\textit{e.g.}, Kullback–Leibler divergence). Secondly, it provides \textit{invariance in the number of equidistant agents}, meaning that the measured diversity is not impacted by the system size when all agents are behaviorally equidistant. This allow to devise diversity control inputs that are independent of the number of agents.

\section{Problem Formulation}
The goal of this work is to control the agents' behavioral diversity $\mathrm{SND}(\left\{\pi_i\right\}_{i\in\mathcal{N}})$ given a desired input diversity $\mathrm{SND}_\mathrm{des}$.
In other words, given a POMG representing a multi-agent task, we are interested in learning a set of policies $\left\{\pi_i\right\}_{i\in\mathcal{N}}$ such that:
\begin{equation}
    \mathrm{SND}(\left\{\pi_i\right\}_{i\in\mathcal{N}}) = \mathrm{SND}_\mathrm{des}.
\end{equation}
To train the policies, we consider actor-critic algorithms (both on and off policy), using the Centralized Training Decentralized Execution (CTDE) MARL paradigm~\cite{lowe2017multi}, where agents are trained using global information and then deployed with independent local policies.  We consider policies represented either as deterministic functions $\pi_i = [\mu_i]$ or as stochastic normal distributions $\pi_i = [\mu_i,\sigma_i]$ with $\mu_i,\sigma_i \in \R^m$. 

\section{Method}
\begin{figure*}[t]
\begin{center}
\centerline{\includegraphics[width=\textwidth]{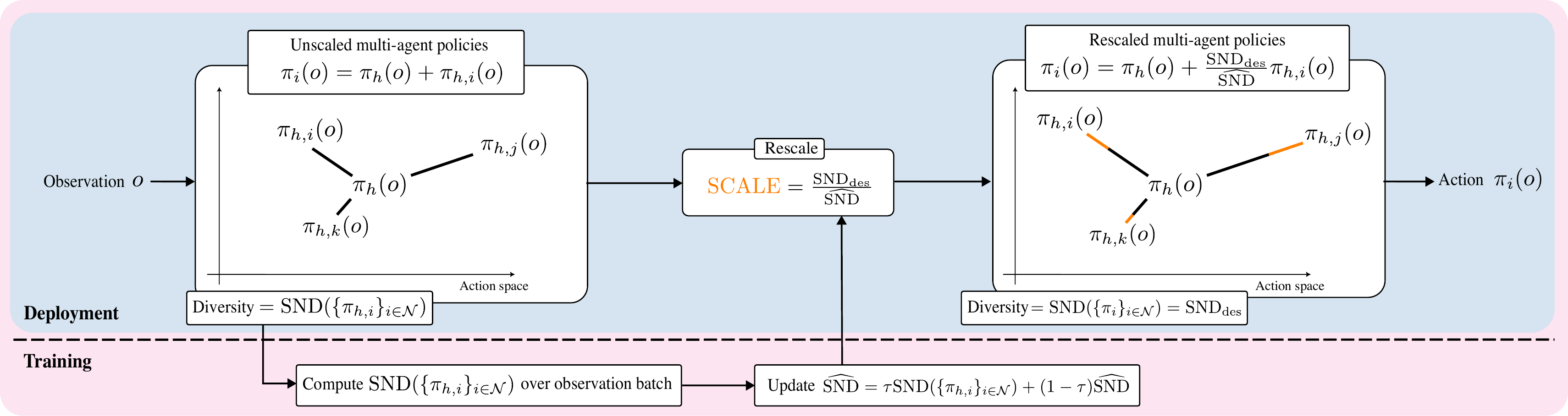}}
\caption{DiCo architecture overview. Multi-agent policies are rescaled to match the desired behavioral diversity $\mathrm{SND}_\mathrm{des}$. The scaling factor is computed as the desired diversity divided by the actual diversity of the unscaled policies, which is updated during training. This process is described in \autoref{alg:cbd}.}
\label{fig:cbd_overview}
\end{center}
\end{figure*}

In this section, we introduce the Diversity Control (DiCo) framework (\autoref{fig:cbd_overview}).
DiCo controls the behavioral diversity of a multi-agent system to a desired value by constraining the agent policies, thereby not requiring any additional intrinsic reward or diversity objective. 
By consisting only of architectural constraints applied to the actor network, it can be utilized with any actor-critic MARL algorithm with  (deterministic or stochastic) continuous actions. 
The DiCo pseudocode is reported in \autoref{app:cbd_pseudocode}.

\subsection{Representing Policies as Heterogeneous Deviations From A Homogeneous Reference}

Multi-agent policies $\pi_i$ have been traditionally represented either as parameter-shared networks or as completely independent functions. Parameter sharing
forces all policies to be identical and thus conditioned on the same set of parameters. This improves the training sample efficiency, but forces behavioral homogeneity (i.e. $\mathrm{SND}(\left \{ \pi_i \right \}_{i \in \mathcal{N}})=0$). On the other hand, completely independent networks can learn heterogeneous functions, but are less sample efficient, as $n$ different policies need to be learned instead of one. 
To leverage the benefits of both approaches and to apply the structural heterogeneity constraint proposed in this work, we consider policies represented by the sum of a \textit{homogeneous parameter-shared component} $\pi_h(o)$ and a \textit{heterogeneous per-agent deviation} $\pi_{h,i}(o)$:
\begin{equation}
    \label{eq:unscaled_policies}
    \pi_i(o) = \pi_h(o) + \pi_{h,i}(o),
\end{equation}
where $\pi_h$ is parameterized by shared parameters $\theta_h$ and $\pi_{h,i}$ is parameterized by per-agent parameters $\theta_{h,i}$. Note that this sum is computed over the action distribution parameters ($\pi_i = [\mu_i,\sigma_i]$), and not the sampled actions. 

By representing multi-agent policies in this manner, we can leverage the sample efficiency benefits of parameter sharing, while avoiding its behavioral homogeneity constraints. 

\subsection{Constraining Heterogeneous Policies via Rescaling}
\label{sec:scaling}
The core idea presented in this work is that of rescaling multi-agent policies (represented as heterogeneous deviations
from a homogeneous reference) to achieve the desired behavioral diversity. To do so we: (1) normalize the agent policy deviations by their diversity prior to rescaling $\widehat{\mathrm{SND}} \coloneqq \mathrm{SND}(\{\pi_{h,i}\}_{i\in\mathcal{N}})$ and (2) multiply them by the desired diversity $\mathrm{SND}_\mathrm{des}$. This yields the following formulation for the agent policies:
\begin{equation}
    \label{eq:scaled_policies}
    \pi_i(o) = \pi_h(o) + \frac{\mathrm{SND}_\mathrm{des}}{\widehat{\mathrm{SND}}}\pi_{h,i}(o).
\end{equation}

\begin{theorem}[Controlling diversity by rescaling agent policies]
\label{thm:scaling}
Given a set of multi-agent policies $\{\pi_i\}_{i\in\mathcal{N}}$ of the form presented in \autoref{eq:scaled_policies} and a desired diversity $\mathrm{SND}_\mathrm{des}$, then
the diversity of the policies $\{\pi_{i}\}_{i\in\mathcal{N}}$ is equal to the desired value: $ \mathrm{SND}(\{\pi_{i}\}_{i\in\mathcal{N}}) = \mathrm{SND}_\mathrm{des}$.
\end{theorem}
\begin{proof} 
We provide proofs for three policy types: (1) deterministic policies $\pi_h(o) = [\mu_{h}(o)]$, $\pi_{h,i}(o) = [\mu_{h,i}(o)]$, (2) policies producing multivariate normal distributions with a homogeneous standard deviation $\pi_h(o) = [\mu_{h}(o),\sigma_h(o)]$, $\pi_{h,i}(o) = [\mu_{h,i}(o),0]$, (3) policies producing multivariate normal distributions with a heterogeneous standard deviation $\pi_h(o) = [\mu_{h}(o),0]$, $\pi_{h,i}(o) = [\mu_{h,i}(o),\sigma_{h,i}(o)]$. The proofs are reported in \autoref{sec:proofs}.
\end{proof}

\autoref{thm:scaling} states that by rescaling multi-agent policies according to \autoref{eq:scaled_policies}, we are guaranteed to achieve the desired diversity $\mathrm{SND}_\mathrm{des}$. Note that in the special case where agents are controlled to be homogeneous ($\mathrm{SND}_\mathrm{des}=0$), the policy formulation becomes a simple parameter-shared network $\pi_i(o) = \pi_h(o)$.

\textbf{Computing $\widehat{\mathrm{SND}}$.}
In order to compute the scaling factor, we need to measure the SND of the policies prior to rescaling:
\begin{align*}
   \widehat{\mathrm{SND}}& \coloneqq  \mathrm{SND}(\left \{ \pi_{h,i} \right \}_{i \in \mathcal{N}})\\ 
   & = \frac{2}{n(n-1)|O|}\sum_{i=1}^n\sum_{j=i+1}^n \sum_{o\in O }W_2(\pi_{h,i}(o),\pi_{h,j}(o)).
\end{align*}
Since this quantity is dependent on the learned policy deviations $\pi_{h,i}$, it needs to be computed only at training time for every policy update.
In doing so, the choice of the observation set $O$, where the diversity is evaluated, plays an important role. $O$ needs to be large enough to accurately represent the observation distribution seen by the agents, while remaining small enough to avoid unnecessary computational overhead.
We construct $O$ from the training data batches\footnote{This is the data sampled from the replay buffer in off-policy algorithms or the training SGD minibatches in on-policy algorithms.}.  The set $O$ is created by unifying the observations from all agents and all timesteps in a batch. This grants an unbiased estimation of the observation distribution given the current policies. To reduce the variance of this estimation, and to stabilize the diversity updates, we further employ a soft update mechanism:
$$
\widehat{\mathrm{SND}} = \tau\mathrm{SND}(\left \{ \pi_{h,i} \right \}_{i \in \mathcal{N}}) +(1-\tau)\widehat{\mathrm{SND}},
$$
where $\tau \in [0,1]\subset\R$ regulates the reliance on the diversity measured from the most recent batch. In \autoref{app:tau_comparison}, we perform an empirical evaluation for different values of $\tau$ to confirm this claim.

\textbf{Determining $\mathrm{SND}_\mathrm{des}$.}
The desired diversity $\mathrm{SND}_\mathrm{des}$ is the main control input required for this method.
Its values represents the desired average Wasserstein distance among agent policy pairs.
While in traditional methods the choice of sharing policy parameters allows a binary decision between a diversity of zero and an unconstrained diversity, $\mathrm{SND}_\mathrm{des}$ allows to constrain the agent policy networks to any specified diversity.
By constraining the policies at a given diversity, the search space of learnable policies is significantly reduced, aiding in the search of the optimal set.
There are a variety of ways to find the value of $\mathrm{SND}_\mathrm{des}$ that will optimize a given problem.
It could be determined in an outer optimization loop, where training is iteratively performed at different diversity levels to optimize a given metric of interest (\textit{e.g.}, performance, resilience)---\autoref{app:closed_control} illustrates an algorithm implementing this closed-loop paradigm. It can also be utilized as a human input prior, decided depending on the multi-agent task. Example applications include: boosting exploration, avoiding local minima in the policy search, or simply inspecting emergent strategies at different diversity levels. In \autoref{sec:experiments} we provide experiments highlighting these applications.  
To aid in the choice of $\mathrm{SND}_\mathrm{des}$, in \autoref{app:max_snd}, we introduce the optimization problem to compute the maximum possible $\mathrm{SND}$ in a bounded action space, which can be solved to obtain an upper bound on the value of $\mathrm{SND}_\mathrm{des}$.

In \autoref{app:method} we present further details about the method, including how to disincentivize the placement of diversity outside the action domain in the context of bounded action spaces.

\section{Case Study: \textit{Multi-Agent Navigation}}
\label{sec:case_study}

\begin{figure*}[t]
\begin{center}
\centerline{\includegraphics[width=\textwidth]{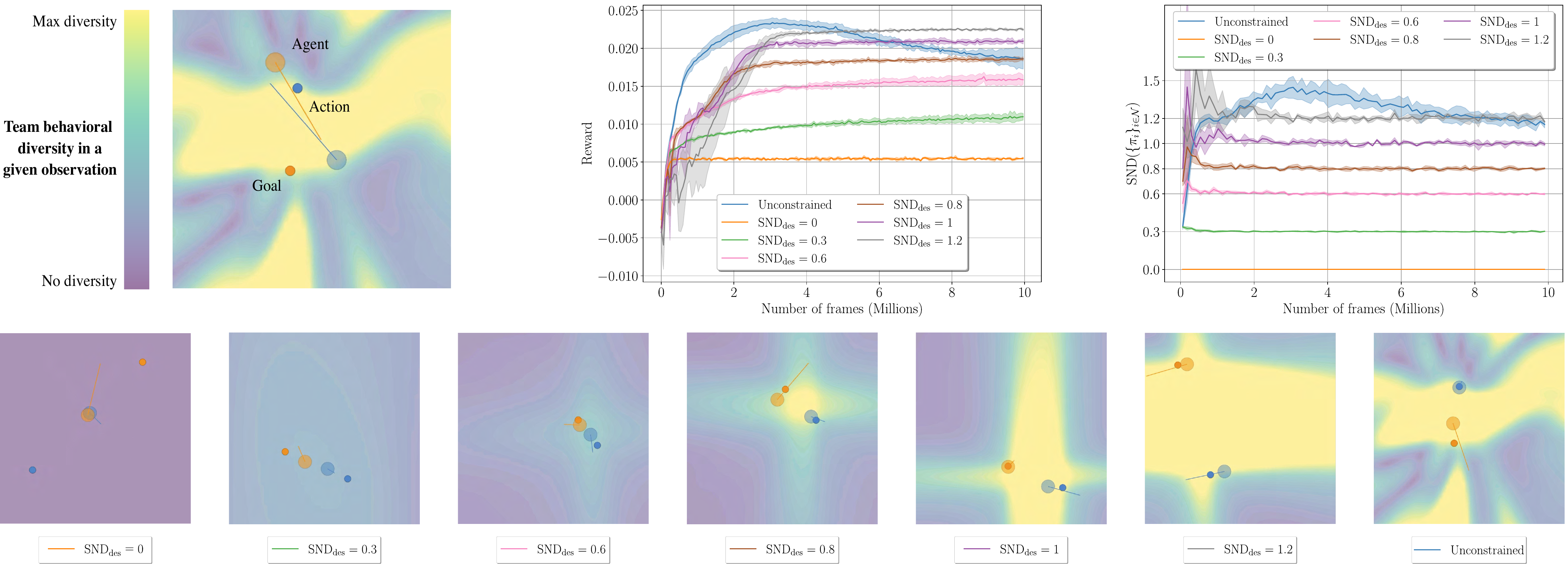}}
\caption{Illustration of the \textit{Multi-Agent Navigation} case study. \textbf{Top left:} Example task rendering illustrating task components.  System diversity is evaluated for each observation in the 2D space and plotted in the background colormap. \textbf{Top center:} Mean instantaneous reward for agents trained with different desired diversities. \textbf{Top right:} $\mathrm{SND}(\left \{ \pi_{i} \right \}_{i \in \mathcal{N}})$ evaluated  for agents trained with different desired diversities. \textbf{Bottom:} Renderings of the diversity distribution over the observation space for agents trained with different desired diversities. With a low diversity budget, agents are not able to go to different goals and learn to converge to the midpoint between goals. As the diversity budget increases, agents learn to distribute diversity in the observations where it is most useful and learn more regular diversity landscapes than the unconstrained case. Curves report mean and standard deviation for the IPPO algorithm over 4 training seeds.}
\label{fig:nav_case_study}
\end{center}
\end{figure*}

To illustrate how DiCo works, we consider the \textit{Multi-Agent Navigation} task (\autoref{fig:nav_case_study}) from the \href{\vmas}{VMAS} simulator~\cite{bettini2022vmas}. In this task, $n=2$ agents are spawned at random positions in a 2D workspace. Each agent is assigned a goal, also spawned at random. Agents observe the relative position to all goals and output a 2D action distribution representing their force movement vector. The reward for each agent is the difference in the relative distance to its goal over consecutive timesteps, incentivizing agents to move towards their goals. 
This task requires heterogeneous behavior as each agent is given the same observation and needs to tackle a different goal.

By comparing multiple policies trained with DiCo, we can analyze the effects of different diversity levels on solution behavior. Using the IPPO algorithm~\citep{de2020independent} in the \href{\benchmarl}{BenchMARL} library, we train policies with various desired diversity levels $\mathrm{SND}_{\mathrm{des}}$, as well as an unconstrained policies (\autoref{eq:unscaled_policies}). The results are reported in  \autoref{fig:nav_case_study}.

In the diversity plot in the top right, we report the $\mathrm{SND}(\left \{ \pi_{i} \right \}_{i \in \mathcal{N}})$ measured throughout learning. Note that all the constrained agents match the desired diversity, providing empirical evidence to support our theoretical claims about DiCo's ability to constrain diversity. In this particular task, we find that performance increases as a function of input diversity, as shown in the reward plot in the top center. This is due to the fact that the task requires heterogeneous behavior and, thus, constraining to a low diversity prevents agents from navigating to different goals. In the limit case where agents are fully homogeneous (\textit{i.e.} $\mathrm{SND}_\mathrm{des}=0$), the agents learn to settle at the midpoint between goals.

To analyze how the agents learn to distribute the diversity budget over the observation space, we define the diversity in a given observation as:
\begin{multline*}
\mathrm{SND}_o(\left \{ \pi_{i} \right \}_{i \in \mathcal{N}},o) =\\ \frac{2}{n(n-1)}\sum_{i=1}^n\sum_{j=i+1}^nW_2(\pi_i(o),\pi_j(o)),
\end{multline*}
and plot it for every $o$ in the observation space.
In other words, for each position in the environment, we plot the diversity between all agent action distributions, if they were computed for that position. 
This is visualized as a colormap overlaid upon the rendered environment, as shown in the bottom of \autoref{fig:nav_case_study}.
They show that, with a low diversity budget, agents are forced to act homogeneously and, thus, all learn to converge between the goals, minimizing the team distance from all goals. As the budget increases, agents are able to take different actions. 
It is important to note that they always learn to place high diversity between the goals, as this area contains the observations where they need to travel opposite directions. Furthermore, they avoid placing heterogeneity in regions where the goals lie in the same direction, as they need to take homogeneous actions in those areas. 
This leads to the emergence of a certain regularity in the diversity distributions of constrained agents, which resemble the shape of a cross. In contrast, unconstrained heterogeneous agents obtain a more chaotic landscape, sometime placing diversity in unnecessary areas.

This case study illustrates the functionality of the DiCo method, showing the emergent diversity distributions over the observation space at different constraint levels. 
In \autoref{app:case_study}, we delve into more detail about this task, discussing the effects of choosing an $\mathrm{SND}_\mathrm{des}$ that is ``too high'', and presenting results from an additional version of this task where all agents need to navigate to the same goal. In such a scenario, where diversity is detrimental, agents with different heterogeneity constraints are forced to learn diverse emergent strategies to achieve the same homogeneous objective. Renderings are available on the \href{\website}{website}\footref{link:experiments}.

\section{Experiments}
\label{sec:experiments}

We present further experiments that highlight different applications of DiCo. We consider the tasks in \autoref{fig:tasks} from the \href{\vmas}{VMAS} simulator~\cite{bettini2022vmas}. Training is performed in the \href{\benchmarl}{BenchMARL} library~\cite{bettini2023benchmarl} using \href{https://github.com/pytorch/rl}{TorchRL}~\cite{bou2023torchrl} in the backend. Depending on the task, we use either the IPPO~\cite{de2020independent} algorithm with stochastic policies,  
or DDPG-based algorithms (i.e., MADDPG~\cite{lowe2017multi}, Independent DDPG (IDDPG)), with deterministic policies. The algorithms have been chosen according to the task needs (\textit{e.g.}, reward sparsity), but DiCo is not limited to these choices and can be applied to any actor-critic algorithm\footnote{All the actor-critic algorithms already available in BenchMARL work out-of-the-box with the DiCo implementation provided (\textit{e.g.}, MASAC, ISAC).}.  

\begin{figure*}[t]
    \newcommand{\subfigsize}{0.15}
    \centering
    \subfigure[\textit{Dispersion}.]{
        \includegraphics[width=\subfigsize\textwidth,frame]{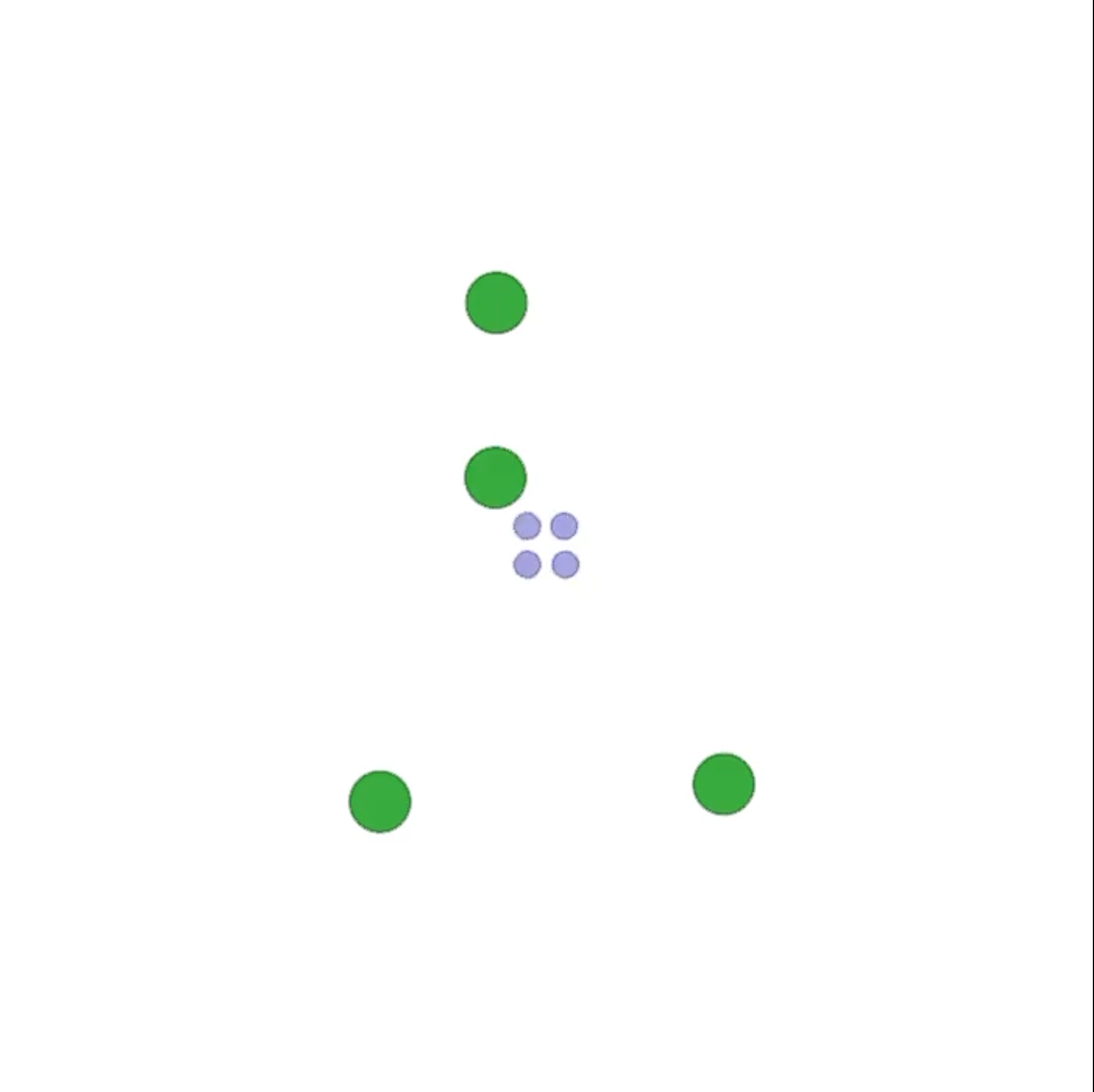}
        \label{fig:dispersion}
    }
    \hfill
    \subfigure[\textit{Sampling}.]{
        \includegraphics[width=\subfigsize\textwidth,frame]{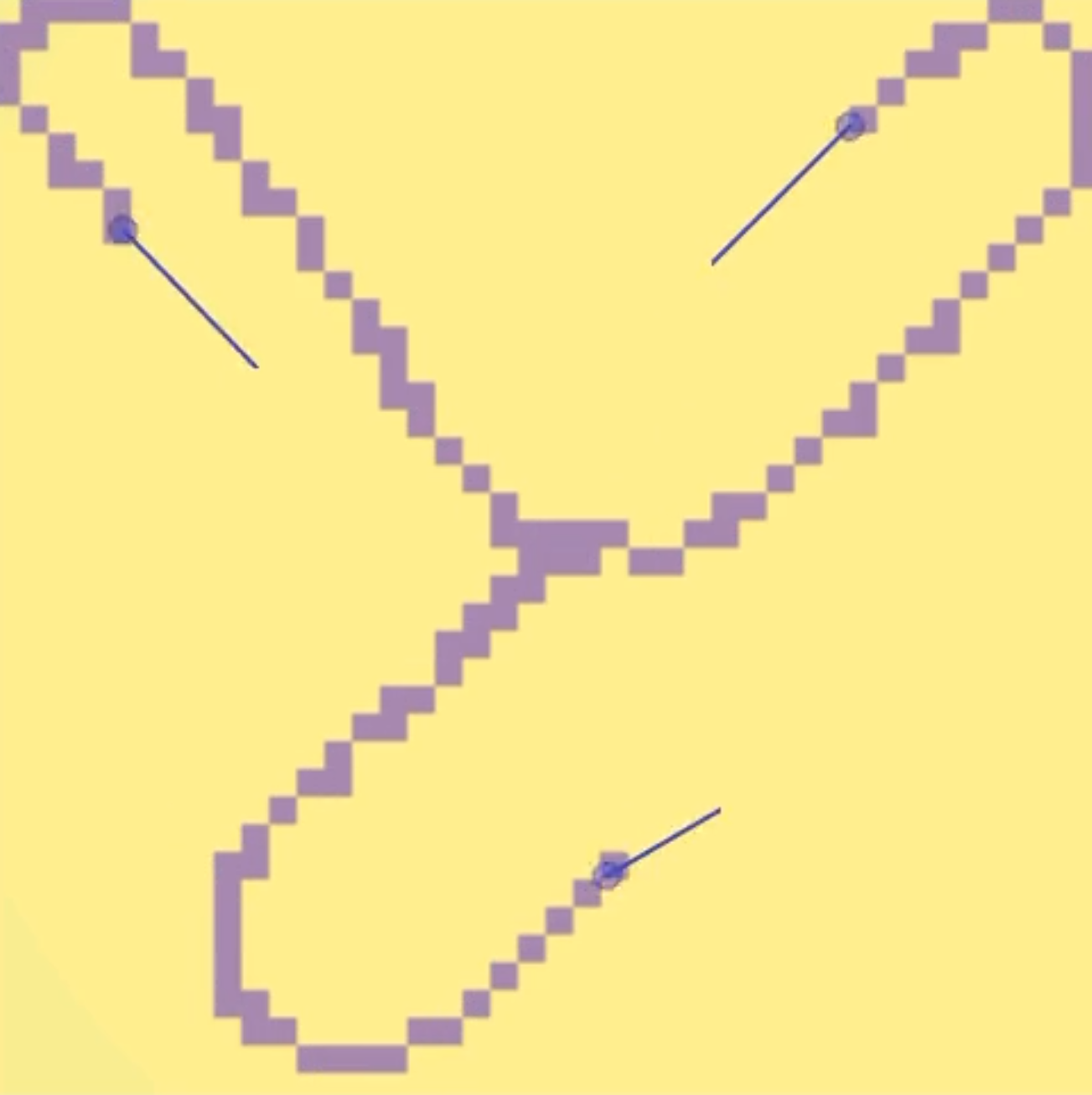}
        \label{fig:sampling}
    }
    \hfill
    \subfigure[\textit{Tag}.]{
        \includegraphics[width=\subfigsize\textwidth,frame]{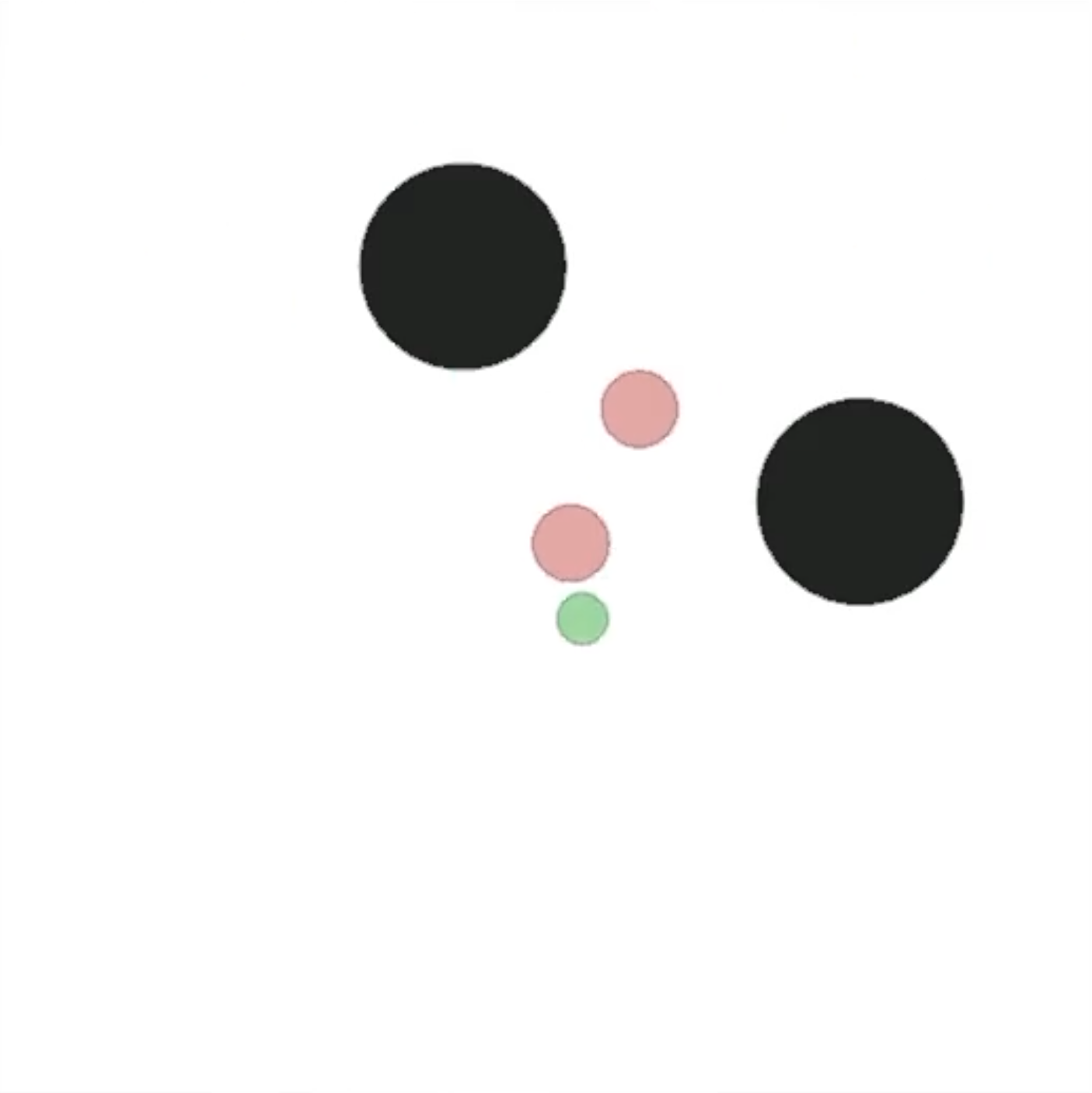}
        \label{fig:tag}
    }
    \hfill
    \subfigure[\textit{Reverse Transport}.]{
        \includegraphics[width=\subfigsize\textwidth,frame]{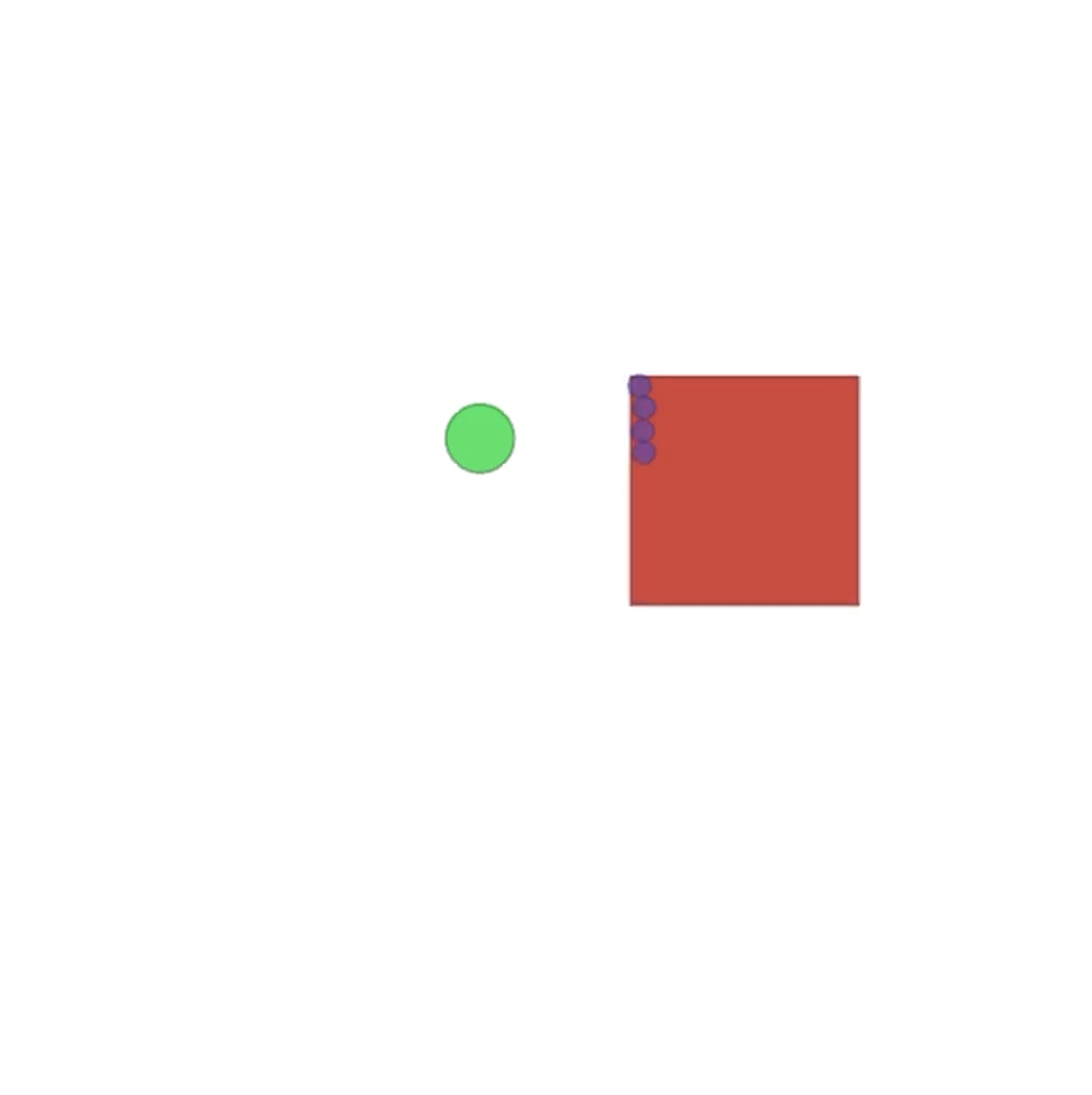}
        \label{fig:rev_transport}
    }
    \caption{Multi-agent tasks from the VMAS simulator analyzed in our experiments.}
    \label{fig:tasks}
\end{figure*}

\subsection{\textit{Dispersion}: Tackling Multiple Objectives}
\label{sec:dispersion}

\begin{figure}[t]
\begin{center}
\centerline{\includegraphics[width=\linewidth]{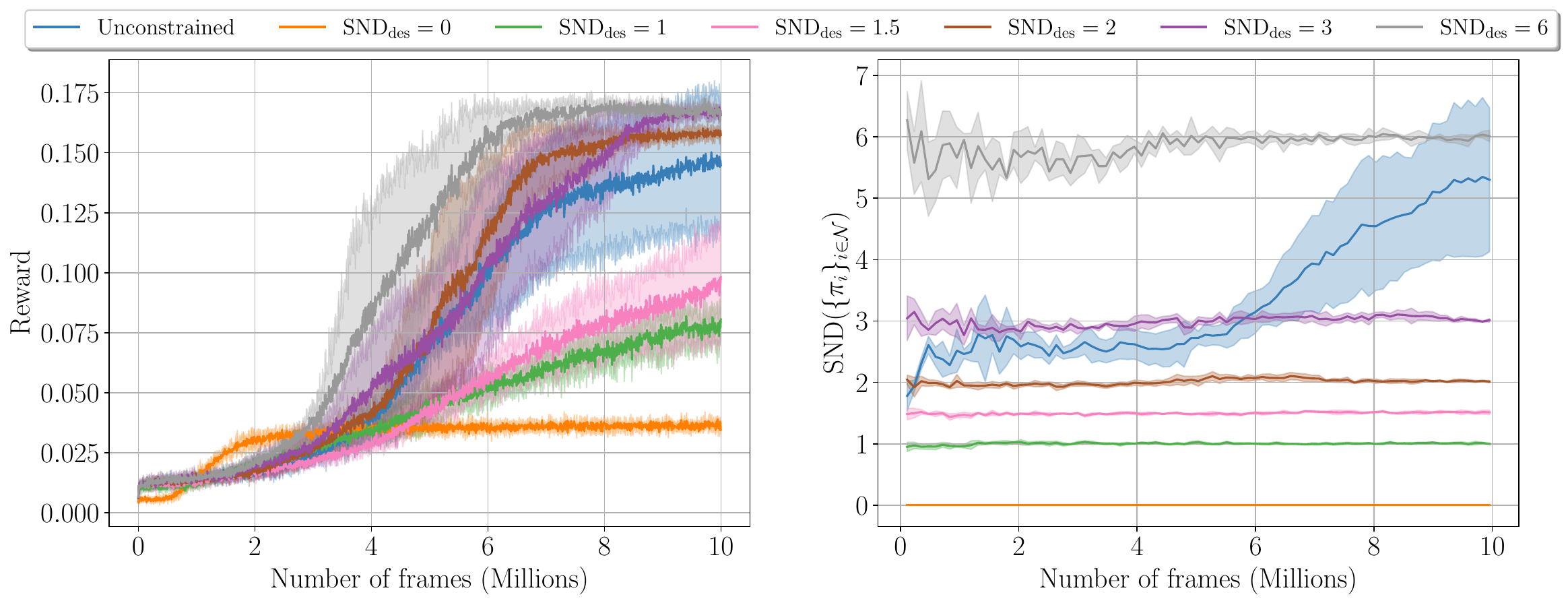}}
\caption{Results from training agents with different constraints on the \textit{Dispersion} task. \textbf{Left:} Mean instantaneous reward. \textbf{Right:} Measured diversity $\mathrm{SND}(\left \{ \pi_{i} \right \}_{i \in \mathcal{N}})$. 
Curves report mean and standard deviation for the MADDPG algorithm over 4 training seeds.}
\label{fig:dispersion_plots}
\end{center}
\end{figure}

In \textit{Dispersion} (\autoref{fig:dispersion}), $n$ agents are spawned in the center of a 2D workspace, and $n$ food particles are spawned around them at random positions. All agents observe the relative position to all food particles and whether they have been consumed. Agent actions are 2D forces that determine their motion. Agents get a sparse reward of 1 for eating food. Rewards are shared among all agents. Each food particle is consumable only once. The optimal policy in this scenario requires each agent to tackle a different food particle.

We train $n=4$ agents using MADDPG with different DiCo constraints, as well as unconstrained heterogeneous polices. In \autoref{fig:dispersion_plots}, we report the mean instantaneous training reward and the measured diversity. As expected, homogeneous agents ($\mathrm{SND}_\mathrm{des} = 0$) achieve
the lowest reward. Having to share the same policy and starting in the same position, homogeneous agents are not able to go to different food particles and thus learn to all tackle one particle at a time in a group. Unconstrained heterogeneous agents, on the other hand, are able to navigate to different food particles, but need to learn the optimal assignment through training (shown by the increasing diversity curve). When training ends, their policy is still suboptimal, with two agents sometimes pursuing the same food. In contrast, by constraining agents to a higher diversity ($\mathrm{SND}_\mathrm{des} = 6$), we are able to bootstrap the diversity discovery process and learn the optimal policy, achieving faster training convergence. 
The results confirm that, in this task, diversity is proportional to performance and we can use DiCo to find higher performing policies faster than unconstrained heterogeneous agents.
\autoref{app:dispersion} further illustrates these results by providing an analysis of the agent trajectories.

\subsection{\textit{Sampling}: Boosting Exploration}
\label{sec:sampling}

\begin{figure}[t]
\begin{center}
\centerline{\includegraphics[width=\linewidth]{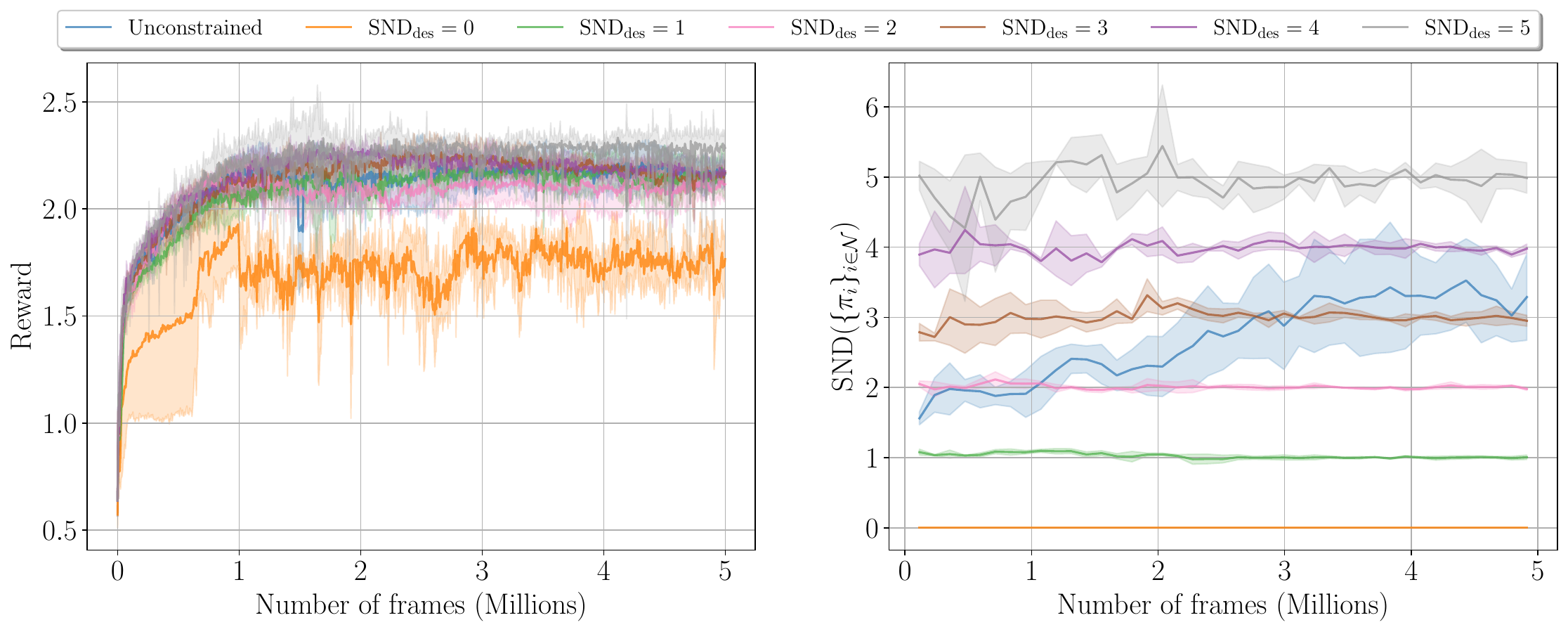}}
\caption{Results from training agents with different constraints on the \textit{Sampling} task. \textbf{Left:} Mean instantaneous reward. \textbf{Right:} Measured diversity $\mathrm{SND}(\left \{ \pi_{i} \right \}_{i \in \mathcal{N}})$. 
Curves report mean and standard deviation for the IDDPG algorithm over 3 training seeds.}
\label{fig:sampling_plots}
\end{center}
\end{figure}

In \textit{Sampling} (\autoref{fig:sampling}), $n$ agents are spawned in the center of a 2D workspace. The workspace presents an underlying uniform distribution (discretized in cells of agent size). Agent actions are 2D forces that determine their motion. Agents observe their position and the value of the distribution in a local $3\times3$ neighborhood. When an agent visits a cell, it takes a sample without replacement, which dictates the reward. Rewards are shared among all agents. Therefore, agents need to spread out and actively sample different parts of the workspace to solve this task. 

We train $n=3$ agents using IDDPG with different DiCo constraints, as well as unconstrained heterogeneous polices. In \autoref{fig:sampling_plots}, we report the mean instantaneous training reward and the measured diversity. 
As in \textit{Dispersion}, this task requires agents to be diverse in order to spread over the observation space. 
Therefore, we see how homogeneous agents ($\mathrm{SND}_\mathrm{des} = 0$) achieve the lowest reward, all visiting the same area of the space. Unconstrained heterogeneous agents are able to spread their coverage, but sometimes revisit cells sampled by others, making their policy suboptimal. By forcing the agents to a higher diversity $\mathrm{SND}_\mathrm{des} = 5$, we are able to enforce a wider spread, with each agent focusing on a separate area of the space. This is shown in the training curves, which indicate that the obtained reward is proportional to the desired diversity.
The results demonstrate that DiCo can be used as an effective method to boost exploration and improve the agent coverage of the observation space, leading to higher performing policies than the unconstrained heterogeneous method.
\autoref{app:sampling} further illustrates these results by providing an analysis of the agent trajectories.

\subsection{\textit{Tag}: Emergent Adversarial Strategies}
\label{sec:tag}

\begin{figure}[t]
\begin{center}
\centerline{\includegraphics[width=\linewidth]{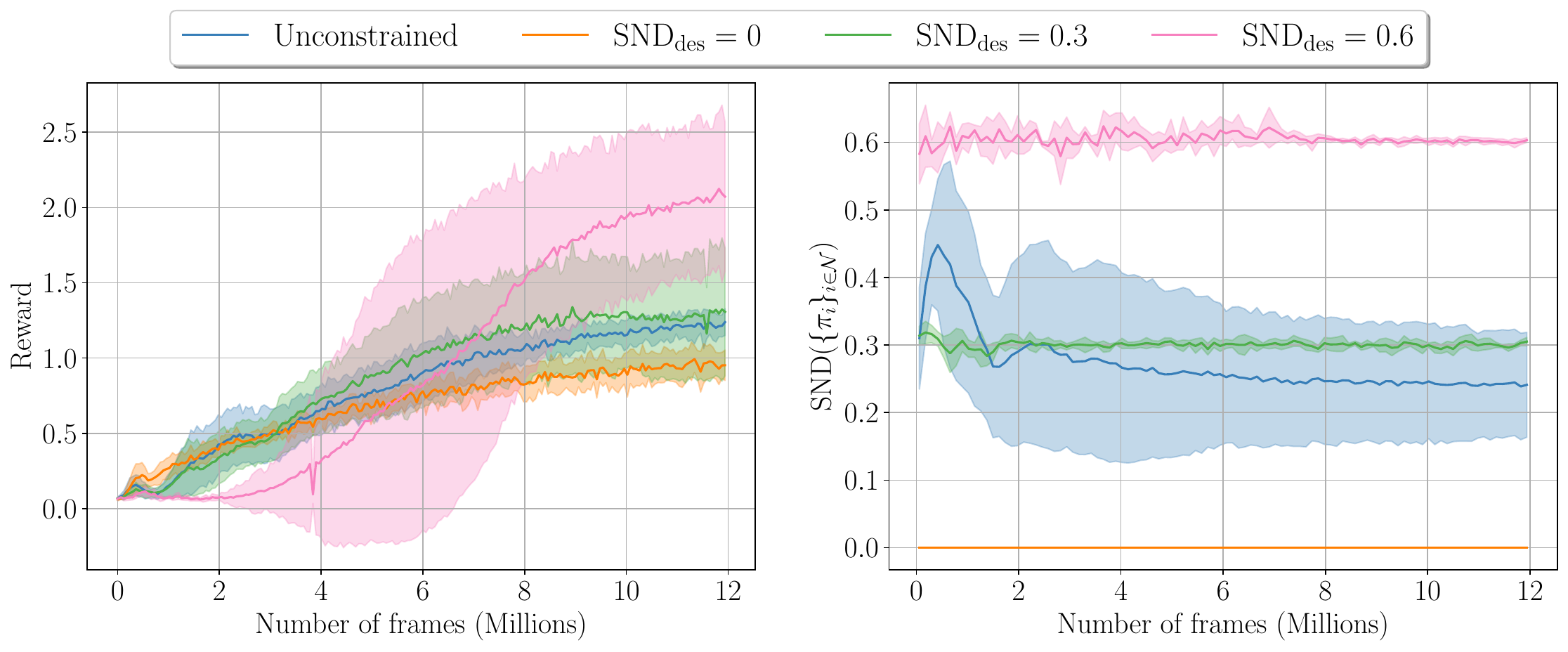}}
\caption{Results from training agents with different constraints on the \textit{Tag} task. \textbf{Left:} Mean instantaneous reward. \textbf{Right:} Measured diversity $\mathrm{SND}(\left \{ \pi_{i} \right \}_{i \in \mathcal{N}})$. 
Curves report mean and standard deviation for the IPPO algorithm over 5 training seeds.}
\label{fig:tag_plots}
\end{center}
\end{figure}

\textit{Tag}~\cite{lowe2017multi}(\autoref{fig:tag}) is an adversarial task where a team of $n$ red agents is rewarded for tagging the green agent. All agents are spawned at random positions in a 2D space, alongside two randomly spawned obstacles (black circles). Agents can collide among each other and with obstacles.
Actions are 2D forces that determine their motion.
They observe their position and the relative positions to all other agents.
Red agents obtain a shared reward of 1 for each timestep where they are touching the green agent. The green agent gets a reward of -1 for the same condition. 

We train $n=2$ red agents using IPPO with different DiCo constraints, as well as unconstrained heterogeneous polices. The green agent is also trained with IPPO for 3 million frames, and frozen for the remainder of training.
In \autoref{fig:tag_plots}, we report the mean instantaneous training reward and the measured diversity of the red agents. From the plots, we observe that homogeneous agents ($\mathrm{SND}_\mathrm{des} = 0$) and unconstrained heterogeneous agents obtain similar rewards, with the heterogeneous model performing slightly better. We also train a DiCo constrained model at the diversity level obtained by the unconstrained one ($\mathrm{SND}_\mathrm{des} = 0.3$) and observe that this model achieves the same reward. By inspecting the rollouts of these policies, we note that they present a behavioral similarity: the red agents blindly chase the green one, all trying to minimize their distance to the target. In real-life ball sports, this is a well-known myopic and suboptimal strategy (\textit{e.g.}, an entire team chasing the ball in soccer). We can intuitively also see its suboptimality in this task: due to the shared nature of the tagging reward, the chasing agents could improve their spatial coverage by diversifying their strategies. To confirm our hypothesis, we perform experiments with a higher desired diversity ($\mathrm{SND}_\mathrm{des} = 0.6$). The results prove our hypothesis, with the constrained model able to almost double the obtained reward. The agents, constrained at this diversity, show the emergence of several fascinating new strategies that resemble strategies employed by human players in ball games (\textit{e.g.}, man-to-man marking, pinching maneuvers, spreading to cut off the evader). See \autoref{app:tag} for an in-depth analysis.
The results demonstrate that, by constraining the policy search space to a specified diversity level, DiCo can be used to learn novel and diverse strategies that can overcome the suboptimality of unconstrained heterogeneous agents.
\autoref{app:tag} analyzes agent rollouts from this experiment, illustrating some of the emergent strategies.

\section{Discussion and Limitations}
\textbf{Lower Diversity Constraints.} Up to this point, we have mainly shown applications where diversity constraints that force a \textit{higher} diversity than the one achieved by unconstrained policies can be used to obtain higher performance. In \autoref{sec:rev_transport}, we provide a further experiment, in the \textit{Reverse Transport} (\autoref{fig:rev_transport}) task, showing that, by leveraging user priors about the role of diversity in a task, it is also possible to improve the training process by forcing a \textit{lower} diversity than the one achieved by the unconstrained method.

\textbf{Inequality Diversity Constraints.} In this paper, we have focused on constraining the agents' diversity to a particular value. However, using the same paradigm, it is also possible to apply more complex constraints. To define constraints in the most general form, we can introduce a function $g: \mathbb{R} \to \mathbb{R}$ that takes in the diversity prior to rescaling $\widehat{\mathrm{SND}}$ and produces the desired diversity fed as input to DiCo. In this formulation, we can define the standard constrained version of DiCo with a function that outputs a constant desired diversity $g(\widehat{\mathrm{SND}}) = \mathrm{SND}_\mathrm{des}$. Similarly, we can write the unconstrained version as $g(\widehat{\mathrm{SND}}) = \widehat{\mathrm{SND}}$. Furthermore, we can define additional forms of diversity control by using \textit{any continuous function $g$}. For example, one particularly useful special case is $g(\widehat{\mathrm{SND}}) = \max(\widehat{\mathrm{SND}}, \mathrm{SND}_\mathrm{des})$ and $g(\widehat{\mathrm{SND}}) = \min(\widehat{\mathrm{SND}}, \mathrm{SND}_\mathrm{des})$, which define inequality constraints (greater than and less than, respectively). Placing upper or lower diversity bounds can be extremely useful in practice to reduce the size of the search space or reinforce some desired behavior (\textit{e.g.}, encouraging exploration, or avoiding excessive deviation from a known safe strategy). In practice, it is very easy to adapt DiCo to control diversity within an upper and/or a lower bound. To implement this, it is sufficient to avoid rescaling the policies if their diversity before the scaling process is already in the allowed range.

\textbf{Analytical Diversity Constraints.} 
In \autoref{app:finn} we discuss an alternative version of DiCo that applies constraints in an \textit{analytical} form. Instead of evaluating the diversity of unscaled policies to determine the scaling factor, it guarantees a given diversity over the \textit{entire} observation space using a fixed-integral neural network~\cite{finn}. 

\textbf{Limitations.} The main limitation of DiCo is that it requires a certain degree of domain knowledge in order to determine the appropriate $\mathrm{SND}_{\mathrm{des}}$ for a given task. 
\autoref{app:div_too_high} discusses this limitation, and proposes solutions to mitigate it. In future work, we are interested in tackling the issue with an iterative diversity optimizer, similar to the one outlined in \autoref{app:closed_control}. In such a paradigm, the targeted diversity value provides an additional parameter that can be optimized in the constrained training process.
Lastly, using inequality constraints instead of an exact diversity can also mitigate the problem of choosing the desired diversity as it allows to specify a diversity range instead of a fixed value.

\section{Conclusion}
We introduce a novel paradigm to control behavioral diversity in MARL.
Using this paradigm, it is possible to constrain the policy search space to a desired behavioral diversity level, improving exploration, performance, and sample efficiency, as well as leading to the emergence of novel strategies.
Our proposed method, DiCo, achieves this by representing policies as the sum of a homogeneous component and heterogeneous components, which are dynamically scaled according the the current and desired value of a given diversity metric.
Unlike existing methods that rely on additional objectives and discrete actions, the architectural nature of the DiCo constraints enables its applicability to any actor-critic algorithm with continuous (stochastic or deterministic) actions.
We theoretically prove that DiCo achieves the desired diversity and empirically demonstrate its functionality through a case study, which illustrates how DiCo agents distribute diversity over the observation space. Finally, our experiments show DiCo's potential in several applications, demonstrating how it can be used as a novel tool to improve performance and sample efficiency in MARL.

\section*{Acknowledgments}
This work was supported by Army Research Laboratory (ARL) Distributed and Collaborative Intelligent Systems and Technology (DCIST) Collaborative Research Alliance (CRA) W911NF-17-2-0181 and the European Research Council (ERC) Project 949940 (gAIa).

\section*{Impact Statement}
This paper presents work whose goal is to advance the field of Multi-Agent Reinforcement Learning. There are many potential societal consequences of our work, none which we feel must be specifically highlighted here.



\bibliography{bibliography}
\bibliographystyle{icml2024}

\newpage
\appendix
\onecolumn
\section{Codebase and Links}
Multimedia material accompanying the paper can be found on the paper's website: \url{\website}.

The code and running instructions are available on GitHub at \url{https://github.com/proroklab/ControllingBehavioralDiversity}. Documentation has been written for all project files to ease the reproduction process.

\subsection{Hyperparameters}

Experiment configurations use the \href{\benchmarl}{BenchMARL}~\cite{bettini2023benchmarl} configuration structure, leveraging Hydra~\cite{Yadan2019Hydra} to decouple YAML configuration files from the Python codebase. 

The submitted code files contain thorough instructions on how to install, run, and configure the project.

Configuration parameters can be found in the \lstinline[columns=fixed]{conf} folder, sorted in \lstinline[columns=fixed]{experiment}, \lstinline[columns=fixed]{algorithm}, \lstinline[columns=fixed]{task}, and \lstinline[columns=fixed]{model} sub-folders. Each file has thorough documentation explaining the effect and meaning of each hyperparameter.
\section{Computational Resources Used}

For the realization of this work, several hours of compute resources have been used.
In particular:
\begin{itemize}
    \item For the purpose of rapid experimentation, we estimate 500 compute hours using an NVIDIA GeForce RTX 2080 Ti GPU and an Intel(R) Xeon(R) Gold 6248R CPU @ 3.00GHz.
    \item For the purpose of running the final experiments on multiple seeds, we estimate 2000 HPC compute hours using an NVIDIA A100-SXM-80GB GPU and 32 cores of a AMD EPYC 7763 64-Core Processor CPU @ 1.8GHz .
\end{itemize}

\section{Computational Complexity}
In the following, we report some aspects related to the computational complexity of DiCo during training and execution.

\textbf{Execution.} Deploying a model trained using DiCo has the same computational cost as deploying an unconstrained heterogeneous network. This is due to the fact that the policy scaling factor is only updated during training and it is fixed at execution time. Therefore, a forward pass of the DiCo model (for $n$ agents)  will just require one forward pass of the homogeneous model and one forward pass of each individual agent network, amounting to $n+1$ forward passes. In case of decentralized deployment (e.g., on physically distinct robots), these would be executed in parallel on separate hardware. In case of centralized deployment, our implementation uses \lstinline{torch.vmap} to batch the $n$ agent calls on GPU, speeding up computation.

\textbf{Training.} In addition to the costs observed for execution, during training DiCo also needs to update the current diversity of the unscaled policies. This requires an additional $n$ forward passes of the individual agent networks to compute the agent actions over the set of evaluation observations $O$, and $\frac{n(n-1)}{2}$ calls to the Wasserstein metric to compute the distance between all agent pairs. In total, this amounts to $2n+1$ forward passes and $\frac{n(n-1)}{2}$ calls to Wasserstein (which is a simple closed-form mathematical expression). Both of these calls can be batched in the inputs (the forward passes are already implemented this way using \lstinline{torch.vmap}).

Therefore, for both training and execution, the number of forward passes (which are usually the time bottleneck) scales linearly in the number of agents. During training, the number of calls to Wassertein scales quadratically in the number of agents, as all pairs need to be evaluated.

\section{Additional Related Works}
\label{app:related}

In this section we present additional works related to the topic of the paper.

\textbf{Population diversity.}
Population-based RL considers evolving a population of agents in order to increase exploration~\cite{wu2023qualitysimilar}.
Promoting diversity in such populations has been shown to lead to improved returns~\cite{parker2020effective}.
This approach has successfully been applied to MARL problems~\cite{jaderberg2019human, vinyals2019grandmaster}.
\citet{charakorn2023generating} proposed a training objective that regularizes agents in a population to find solutions that are compatible with their partner agents while not compatible with any other agents in the population. However the proposed approach is limited to 2-player games.
\citet{ingvarsson2023mix} apply the MAP-Elites Quality-Diversity algorithm to MARL to discover populations of diverse and high-preforming teams. All the works presented consider the problem of promoting diversity among agents in a population. In contrast, our work does not train populations of agents and considers the problem of controlling the behavioral diversity of concurrently acting agents in a task.
Further works~\cite{perez2021modelling,liu2021towards,liu2022unified} introduce novel diversity metrics that they then maximize as part of the learning objective, either via the use of additional loss terms or intrinsic rewards. In our literature review, we refer to these approaches as ‘diversity boosting’, and describe their differences from our approach. In particular, the main difference to our work, apart from the type of games tackled and the population-based context, is that DiCo aims to avoid controlling diversity by augmenting the learning objective and instead applies diversity constraints directly to the policy structure.

\section{DiCo with a General Diversity Metric}
\label{app:general_metric}

This section presents a theorem that defines the property that a diversity metric needs to satisfy in order for \autoref{thm:scaling} to hold and the metric being usable with DiCo.

\begin{theorem}[DiCo with general diversity metric]
\label{thm:general_metric}
Given a general diversity metric $\mathrm{M}: \{\pi_i\}_{i\in \mathcal{N}} \mapsto \R_{\geq0}$, if it holds that $\mathrm{M}(\{\pi_h + c\pi_{h,i}\}_{i\in \mathcal{N}}) = c\mathrm{M}(\{\pi_{h,i}\}_{i\in \mathcal{N}})$, for all $c \in \R_{\geq0}$,
then \autoref{thm:scaling} holds for this metric. In particular, when $c = \frac{\mathrm{M}_\mathrm{des}}{\widehat{\mathrm{M}}}$, the diversity of the scaled policies matches the desired diversity: 
$$
\mathrm{M}\left(\left\{\pi_h + \frac{\mathrm{M}_\mathrm{des}}{\widehat{\mathrm{M}}}\pi_{h,i}\right\}_{i\in \mathcal{N}}\right) = \mathrm{M}_\mathrm{des}.
$$
\end{theorem}
\begin{proof}
    By applying the property defined in the theorem to the left hand side of the equation, we get:
    $$
    \frac{\mathrm{M}_\mathrm{des}}{\widehat{\mathrm{M}}}
    \mathrm{M}\left(\left\{ \pi_{h,i}\right\}_{i\in \mathcal{N}}\right) = \mathrm{M}_\mathrm{des}.
    $$
    We know that $\widehat{\mathrm{M}}$ is defined as $\mathrm{M}\left(\left\{ \pi_{h,i}\right\}_{i\in \mathcal{N}}\right)$. Thus, we can simplify it, leaving:
    $$
     \mathrm{M}_\mathrm{des} = \mathrm{M}_\mathrm{des}.
    $$
\end{proof}
This property implies that a diversity metric $\mathrm{M}$ needs:
\begin{itemize}
    \item to not be dependent on the homogeneous policy component (an assumption that we find reasonable, as this component is equal for all agents);
    \item to follow the mathematical property of homogeneity\footnote{\url{https://en.wikipedia.org/wiki/Homogeneous_function}} ($cf(x) = f(cx)$, a property that, for example, holds for all norm functions) (here “homogeneity” refers to the mathematical property and not the agent homogeneity from our paper).
\end{itemize}
Therefore, any diversity measure that is based on the norm of the difference between agent policies will satisfy this assumption. Several metrics satisfy this property, such as: Hierarchical Social Entropy (HSE)~\cite{balch2000hierarchic} and SND with the determinant~\cite{parker2020effective} of the behavioral distance matrix (see \cite{bettini2023snd} for the definition of such matrix) as the system aggregation function (instead of the mean, as in SND).

Users that want to use DiCo with diversity metrics other than SND and HSE, just need to prove that this property holds for their desired metric to be able to run DiCo with it.

\subsection{Motivation for SND}

This paper uses SND as the metric for DiCo. 
Our choice of SND as the metric employed in this work is motivated by various reasons. Most of these reasons are introduced in the SND paper~\cite{bettini2023snd}. In the following, we report the core motivations.

\textbf{Lower level (distance between agents).} For the lower level distance between agent policies we are interested in using a statistical metric that has a closed-form equation for multivariate normal distributions (our policy parametrization). Statistical divergences (i.e. (symmetric) KL divergence, JS-Divergence) do not follow triangular inequality and thus do not have the properties of a metric. Furthermore, (symmetric) KL divergence has a value of infinity for Delta distributions, while Wasserstein has a value equal to the distance between the means. Reinforcement learning policies often converge to Delta distributions after the exploration phase, and, thus, using KL would lead to infinity distances among agents. For these reasons, Wasserstein proves a valid candidate that satisfies all our theoretical and practical requirements.

\textbf{Higher level (system aggregation).} The choice of a higher level aggregation function that provides a single scalar given the matrix of agent distances is a complex one. The purpose of this aggregation is to summarize information about a particular property of the system, which will inevitably incur in some information loss. In the same way that we cannot distinguish $\mathrm{mean}(4,4,1)$ and $\mathrm{mean}(3,3,3)$, it is common for aggregation functions to map different sets to the same value (e.g., $\max(1,0,0)=\max(1,1,0)$, $\mathrm{sum}(1,1,1)=\mathrm{sum}(3,0,0)$). Therefore, the choice of the aggregation function depends on the type of information loss that we would like to incur. 
The main reason why we adopt the mean aggregation from SND is due to its two core properties presented in~\cite{bettini2023snd}: \textit{Invariance in the Number of Equidistant Agents} and \textit{Measure of Behavioral Redundancy}. In particular, the first property allows us to design the target independently from the number of agents in the task.
\section{Proofs of \autoref{thm:scaling}.}
\label{sec:proofs}

In this section, we will provide the proofs for \autoref{thm:scaling}. We will use the definition of $\mathrm{SND}$ reported in \autoref{eq:snd} and the closed-form solution for computing the Wasserstein metric between multivariate normal distributions. 

\textbf{Wasserstein metric for multivariate normal distributions: }
Let $\pi_1 = \mathcal{N}(\mu_1,\Sigma_1)$ and $\pi_2 = \mathcal{N}(\mu_2,\Sigma_2)$ be two multivariate normal distributions on $\R^m$. Then, the 2-Wasserstein distance between $\pi_1$ and $\pi_2$ is computed as:
\begin{equation}
\label{eq:wasserstein}
W_2(\pi_1,\pi_2) = \sqrt{||\mu_1-\mu_2||^2_2  + \mathrm{trace}(\Sigma_1+\Sigma_2- 2(\Sigma_2^{\frac{1}{2}}\Sigma_1\Sigma_2^{\frac{1}{2}})^{\frac{1}{2}})} 
\end{equation}
In this work, we consider policies with the form $\pi_i = \mathcal{N}(\mu_i,\sigma_i)$, with $\mu_i,\sigma_i\in\R^m$, where $\sigma_i$ is a standard deviation vector which uniquely defines a diagonal covariance matrix $\Sigma_i\in\R^{m\times m}$, $\Sigma_i = \mathrm{diag}(\sigma_i^2)$.

\begin{proof}
Our goal is to prove that:
$$
\mathrm{SND}(\{\pi_{i}\}_{i\in\mathcal{N}}) = \mathrm{SND}_\mathrm{des},
$$
where $\widehat{\mathrm{SND}} = \mathrm{SND}(\{\pi_{h,i}\}_{i\in\mathcal{N}})$ in the policy formulation of \autoref{eq:scaled_policies}.

By substituting the $\mathrm{SND}$ formulation from \autoref{eq:snd} we obtain:
$$
\frac{2}{n(n-1)|O|}\sum_{i=1}^n\sum_{j=i+1}^n \sum_{o\in O }W_2(\pi_i(o),\pi_j(o)) = \mathrm{SND}_\mathrm{des}.
$$
We further substitute the definition of policies $\{\pi_i\}_{i\in\mathcal{N}}$ from \autoref{eq:scaled_policies}:
\begin{equation}
\label{eq:proof_snd}
\frac{2}{n(n-1)|O|}\sum_{i=1}^n\sum_{j=i+1}^n \sum_{o\in O }W_2 \left (\pi_h(o) + \frac{\mathrm{SND}_\mathrm{des}}{\widehat{\mathrm{SND}}}\pi_{h,i}(o),\pi_h(o) + \frac{\mathrm{SND}_\mathrm{des}}{\widehat{\mathrm{SND}}}\pi_{h,j}(o) \right ) = \mathrm{SND}_\mathrm{des}.
\end{equation}

This is the formulation we need to prove.

\textbf{Deterministic policies: }
We consider deterministic polices $\pi_h(o) = [\mu_{h}(o)]$, $\pi_{h,i}(o) = [\mu_{h,i}(o)]$.
We can thus rewrite \autoref{eq:proof_snd} using the $W_2$ formulation from \autoref{eq:wasserstein}. We observe that the covariance term can be removed, due to the deterministic nature of the policies.

\begin{equation*}
\frac{2}{n(n-1)|O|}\sum_{i=1}^n\sum_{j=i+1}^n \sum_{o\in O }
\sqrt{\left \|\left ( \mu_h(o) + \frac{\mathrm{SND}_\mathrm{des}}{\widehat{\mathrm{SND}}}\mu_{h,i}(o)\right )-\left ( \mu_h(o) + \frac{\mathrm{SND}_\mathrm{des}}{\widehat{\mathrm{SND}}}\mu_{h,j}(o)\right ) \right \|^2_2} = \mathrm{SND}_\mathrm{des}.
\end{equation*}

We notice that $\mu_h(o)$ simplifies. Further, the scaling factor $\frac{\mathrm{SND}_\mathrm{des}}{\widehat{\mathrm{SND}}}$ can be taken out of all the equation layers as it is positive and does not depend on $i,j,o$. We can rewrite as:

\begin{equation*}
\frac{\mathrm{SND}_\mathrm{des}}{\widehat{\mathrm{SND}}}\frac{2}{n(n-1)|O|}\sum_{i=1}^n\sum_{j=i+1}^n \sum_{o\in O }
\sqrt{\left \|\mu_{h,i}(o) -  \mu_{h,j}(o) \right \|^2_2} = \mathrm{SND}_\mathrm{des}.
\end{equation*}

We can simplify $\mathrm{SND}_\mathrm{des}$, yielding:

\begin{equation*}
\frac{2}{n(n-1)|O|}\sum_{i=1}^n\sum_{j=i+1}^n \sum_{o\in O }
\sqrt{\left \|\mu_{h,i}(o) -  \mu_{h,j}(o) \right \|^2_2}  = \widehat{\mathrm{SND}},
\end{equation*}

which is true by definition, since we defined $ \widehat{\mathrm{SND}}$ as:

\begin{equation*}
\begin{split}
\widehat{\mathrm{SND}}& = \mathrm{SND}(\{\pi_{h,i}\}_{i\in\mathcal{N}}) \\
& = \frac{2}{n(n-1)|O|}\sum_{i=1}^n\sum_{j=i+1}^n \sum_{o\in O }W_2(\pi_{h,i}(o),\pi_{h,j}(o))
\\ & = \frac{2}{n(n-1)|O|}\sum_{i=1}^n\sum_{j=i+1}^n \sum_{o\in O }
\sqrt{\left \|\mu_{h,i}(o) -  \mu_{h,j}(o) \right \|^2_2} .
\end{split}
\end{equation*}

\textbf{Stochastic multivariate normal policies with a homogeneous standard deviation: }
We consider stochastic policies of the type $\pi_h(o) = [\mu_{h}(o),\sigma_h(o)]$, $\pi_{h,i}(o) = [\mu_{h,i}(o),0]$.

From the definition of the policies in \autoref{eq:scaled_policies}, we obtain $\sigma_{i}(o) =  \sigma_{h}(o) + 0$  and thus $\Sigma_{i}(o)= \Sigma_{h}(o)$. The covariance term in $W_2$ from \autoref{eq:wasserstein} can be written as:

$$
\mathrm{trace}(\Sigma_h+\Sigma_h- 2(\Sigma_h^{\frac{1}{2}}\Sigma_h\Sigma_h^{\frac{1}{2}})^{\frac{1}{2}})= 0.
$$

Thus, it can be removed, and all the same steps from the deterministic policies case can be followed to prove the theorem. The insight behind this proof is that: if the standard deviation term is computed by the parameter-shared network only, it does not contribute to the behavioral diversity of the agents' action distributions.

\textbf{Stochastic multivariate normal policies with a heterogeneous standard deviation: }
We consider stochastic policies of the type $\pi_h(o) = [\mu_{h}(o),0]$, $\pi_{h,i}(o) = [\mu_{h,i}(o),\sigma_{h,i}(o)]$. 

From the definition of the policies in \autoref{eq:scaled_policies}, we obtain $\sigma_{i}(o) = 0+ \frac{\mathrm{SND}_\mathrm{des}}{\widehat{\mathrm{SND}}}\sigma_{h,i}(o)$  and, thus, $\Sigma_{i}(o)= \left(\frac{\mathrm{SND}_\mathrm{des}}{\widehat{\mathrm{SND}}}\right )^2\Sigma_{h,i}(o)$ The covariance term in $W_2$ from \autoref{eq:wasserstein} becomes:

\begin{multline*}
\mathrm{trace}\left (\left (\frac{\mathrm{SND}_\mathrm{des}}{\widehat{\mathrm{SND}}}\right )^2\Sigma_{h,i}(o)+\left (\frac{\mathrm{SND}_\mathrm{des}}{\widehat{\mathrm{SND}}}\right )^2\Sigma_{h,j}(o)\right.\\\left.- 2\left (\left(\left (\frac{\mathrm{SND}_\mathrm{des}}{\widehat{\mathrm{SND}}}\right )^2\Sigma_{h,j}(o)\right )^{\frac{1}{2}}\left (\frac{\mathrm{SND}_\mathrm{des}}{\widehat{\mathrm{SND}}}\right )^2\Sigma_{h,i}(o)\left(\left (\frac{\mathrm{SND}_\mathrm{des}}{\widehat{\mathrm{SND}}}\right )^2\Sigma_{h,j}(o)\right )^{\frac{1}{2}}\right )^{\frac{1}{2}}\right ).
\end{multline*}

Since $\Sigma_{h,i}(o)$ and $\Sigma_{h,j}(o)$ are diagonal matrices, we can take the scaling factor to the outer layer of the equation:

$$
\left (\frac{\mathrm{SND}_\mathrm{des}}{\widehat{\mathrm{SND}}}\right )^2
\mathrm{trace}\left (\Sigma_{h,i}(o)+\Sigma_{h,j}(o)- 2\left (\Sigma_{h,j}(o)^{\frac{1}{2}}\Sigma_{h,i}(o)\Sigma_{h,j}(o)^{\frac{1}{2}}\right )^{\frac{1}{2}}\right ).
$$

By substituting $W_2$ in \autoref{eq:proof_snd} we get:

\begin{multline*}
\frac{2}{n(n-1)|O|}\sum_{i=1}^n\sum_{j=i+1}^n \sum_{o\in O }
\left ( \left \| \left ( \mu_h(o) + \frac{\mathrm{SND}_\mathrm{des}}{\widehat{\mathrm{SND}}}\mu_{h,i}(o)\right )-\left ( \mu_h(o) + \frac{\mathrm{SND}_\mathrm{des}}{\widehat{\mathrm{SND}}}\mu_{h,j}(o)\right ) \right \|^2_2 + \right. 
\\ \left. \left (\frac{\mathrm{SND}_\mathrm{des}}{\widehat{\mathrm{SND}}}\right )^2
\mathrm{trace}\left (\Sigma_{h,i}(o)+\Sigma_{h,j}(o)- 2\left (\Sigma_{h,j}(o)^{\frac{1}{2}}\Sigma_{h,i}(o)\Sigma_{h,j}(o)^{\frac{1}{2}}\right )^{\frac{1}{2}} \right ) \right )^{\frac{1}{2}} = \mathrm{SND}_\mathrm{des},
\end{multline*}

which can be rewritten as:


\begin{multline*}
\frac{\mathrm{SND}_\mathrm{des}}{\widehat{\mathrm{SND}}}
\frac{2}{n(n-1)|O|}\sum_{i=1}^n\sum_{j=i+1}^n \sum_{o\in O }
\left ( \left \| \left ( \mu_h(o) + \mu_{h,i}(o)\right )-\left ( \mu_h(o) + \mu_{h,j}(o)\right ) \right \|^2_2 + \right. 
\\ \left. 
\mathrm{trace}\left (\Sigma_{h,i}(o)+\Sigma_{h,j}(o)- 2\left (\Sigma_{h,j}(o)^{\frac{1}{2}}\Sigma_{h,i}(o)\Sigma_{h,j}(o)^{\frac{1}{2}}\right )^{\frac{1}{2}} \right ) \right )^{\frac{1}{2}} = \mathrm{SND}_\mathrm{des},
\end{multline*}

We can simplify $\mathrm{SND}_\mathrm{des}$, yielding:

\begin{multline*}
\frac{2}{n(n-1)|O|}\sum_{i=1}^n\sum_{j=i+1}^n \sum_{o\in O }
\left (\left \|\mu_{h,i}(o) -  \mu_{h,j}(o) \right \|^2_2 + \right. 
\\ \left.  
\mathrm{trace}\left (\Sigma_{h,i}(o)+\Sigma_{h,j}(o)- 2\left (\Sigma_{h,j}(o)^{\frac{1}{2}}\Sigma_{h,i}(o)\Sigma_{h,j}(o)^{\frac{1}{2}}\right )^{\frac{1}{2}} \right )\right)^\frac{1}{2}  = \widehat{\mathrm{SND}},
\end{multline*}

which is true by definition, since we defined $ \widehat{\mathrm{SND}}$ as:

\begin{equation*}
\begin{split}
\widehat{\mathrm{SND}}& = \mathrm{SND}(\{\pi_{h,i}\}_{i\in\mathcal{N}}) \\
&= \frac{2}{n(n-1)|O|}\sum_{i=1}^n\sum_{j=i+1}^n \sum_{o\in O }W_2(\pi_{h,i}(o),\pi_{h,j}(o))
\\ &= 
\!\begin{multlined}[t]
\frac{2}{n(n-1)|O|}\sum_{i=1}^n\sum_{j=i+1}^n \sum_{o\in O }\left ( \left \|\mu_{h,i}(o) -  \mu_{h,j}(o) \right \|^2_2 + 
\right. \\ \left.   \mathrm{trace}\left (\Sigma_{h,i}(o)+\Sigma_{h,j}(o)- 2\left (\Sigma_{h,j}(o)^{\frac{1}{2}}\Sigma_{h,i}(o)\Sigma_{h,j}(o)^{\frac{1}{2}}\right )^{\frac{1}{2}} \right ) \right )^\frac{1}{2}.
\end{multlined}
\end{split}
\end{equation*}

\end{proof}
\section{Additional Method Details}
\label{app:method}

In this section, we provide further details about the methods.

\subsection{DiCo Pseudocode}
\label{app:cbd_pseudocode}

In \autoref{alg:cbd}, we present the pseudocode for policy evaluation in DiCo. This process is depicted in \autoref{fig:cbd_overview}. In lines 2-4, the estimated diversity is initialized at the desired value. In lines 5-8, the training process recomputes it using a soft update. Lastly, in lines 9-11 the multi-agent outputs are computed over the observation batch according to the policies formulation from \autoref{eq:scaled_policies}.

\begin{algorithm}[ht]
   \caption{Policy evaluation in Diversity Control (DiCo).}
   \label{alg:cbd}
\begin{algorithmic}[1]
   \STATE {\bfseries Input:} observation batch $O$, homogeneous policy  $\pi_h$, heterogeneous policies $\pi_{h,0},\ldots,\pi_{h,n}$, soft-update coefficient $\tau$, desired diversity $\mathrm{SND}_\mathrm{des}$
   \IF{Init}
        \STATE $\widehat{\mathrm{SND}} = \mathrm{SND}_\mathrm{des}$
   \ENDIF
   \IF{Training}
        \STATE Compute $\mathrm{SND}(\left \{ \pi_{h,i} \right \}_{i \in \mathcal{N}})$ over $O$
        \STATE $\widehat{\mathrm{SND}} = \tau\mathrm{SND}(\left \{ \pi_{h,i} \right \}_{i \in \mathcal{N}}) +(1-\tau)\widehat{\mathrm{SND}}$
   \ENDIF
   \FOR{$i \in \mathcal{N}$}
        \STATE $\pi_i(O) = \pi_h(O) + \frac{\mathrm{SND}_\mathrm{des}}{\widehat{\mathrm{SND}}}\pi_{h,i}(O)$
   \ENDFOR
\end{algorithmic}
\end{algorithm}

\subsection{Closed-Loop Diversity Control}
\label{app:closed_control}

DiCo provides a new control input ($\mathrm{SND}_\mathrm{des}$) which constrains multi-agent policies to the desired diversity. However, in tasks where a prior on the desired diversity is not available, determining its value might not be straightforward. Furthermore, since training optimizes only for returns, this additional input could be used to optimize a secondary metric of interest (\textit{e.g.}, resilience).

In this section, we present the pseudocode of an algorithm that can be used to automatically control diversity using a closed-loop PID controller~\citep{ang2005pid} to optimize a metric $M$ of interest, given a desired metric value $M_\mathrm{des}$. This metric can be defined by the user to measure properties of interest of the desired system. 
For example, a user might be interested to train a system achieving a certain success in the face of agent failures or noise. To achieve the desired result, they would design a metric $M$ that reflects the desired properties and use it to optimize $\mathrm{SND}_\mathrm{des}$. The code is reported in \autoref{alg:closed_loop}.

Lines 4-11 iteratively train the system at different diversity levels $\mathrm{SND}_\mathrm{des}$, trying to minimize the error $e = M_{\mathrm{des}}-M$. The PID terms of the controller take into account the error variation  over time and its accumulated value. In line 11, the iteration process is terminated as soon as the error falls
within a predetermined bound $-\epsilon < e < \epsilon$.

\begin{algorithm}[ht]
   \caption{PID control loop to optimize a desired metric using DiCo.}
   \label{alg:closed_loop}
\begin{algorithmic}[1]
   \STATE {\bfseries Input:} metric to optimize $M$, desired metric value $M_{\mathrm{des}}$, initial desired diversity $\mathrm{SND}_{\mathrm{init}}$, error tolerance $\epsilon$, PID control parameters $K_p,K_d,K_i$
   \STATE $\mathrm{SND}_{\mathrm{des}} = \mathrm{SND}_{\mathrm{init}}$
   \STATE $e=e_{\mathrm{accu}}=0$
   \REPEAT
        \STATE Train agents at $\mathrm{SND}_{\mathrm{des}}$ over multiple seeds
        \STATE Compute aggregate metric $M$ over training runs
        \STATE $e_\mathrm{prev} = e$
        \STATE Compute error $e = M_{\mathrm{des}}-M$
        \STATE $e_{\mathrm{accu}}= e+e_{\mathrm{accu}}$ 
        \STATE $\mathrm{SND}_{\mathrm{des}} =\max( K_p e + K_d (e - e_\mathrm{prev})+K_ie_{\mathrm{accu}},0)$

   \UNTIL{$-\epsilon < e < \epsilon$}
\end{algorithmic}
\end{algorithm}

\subsection{Maximum $\mathrm{SND}$ for a Bounded Action Space}
\label{app:max_snd}
$\mathrm{SND}$ is a metric that ranges from $0$ to infinity. When considering bounded action spaces and deterministic actions ($\pi_i=[\mu_i]$), the metric can furtherly be bounded\footnote{In the case of stochastic actions, $\sigma_i(o)$ is not bounded and thus diversity can go to infinity in its growth.}. 

Given an action domain with a lower bound $a_{\mathrm{min}}\in \R^m$, an upper bound $a_{\mathrm{max}}\in \R^m$, and a set of actions $\mu_i\in \R^m, \forall i \in \mathcal{N}$, the maximum $\mathrm{SND}$ is the solution to the following convex problem:
\begin{center}
\begin{equation*}
        \max_{\{\mu_i\}_{i\in \mathcal{N}}}  \quad \frac{2}{n(n-1)} \sum_{i=1}^n  \sum_{j=i+1}^n\left \| \mu_i-\mu_j\right \|_2 
\end{equation*}
s.t.
\begin{equation*}
        a_{\mathrm{min}}^k \leq \mu_i^k \quad \forall k \in [1,\ldots,m], \forall i \in\mathcal{N}
\end{equation*}
\begin{equation*}
        \mu_i^k  \leq a_{\mathrm{max}}^k  \quad  \forall k \in [1,\ldots,m], \forall i \in\mathcal{N}.
\end{equation*}
\end{center}

This problem can be solved using open-sourced and free to use convex optimization solvers, such as CVXPY~\cite{diamond2016cvxpy}.

\subsection{Disincentivizing Diversity Outside the Action Space}

In the context of bounded action spaces $[a_{\mathrm{min}},a_{\mathrm{max}}]$ with $a_{\mathrm{min}},a_{\mathrm{max}}\in \R^m$, and deterministic actions, increasing the diversity of
 the rescaled policies from \autoref{eq:scaled_policies} might lead to action means $\mu_i$ outside the action space boundaries.
Since actions are clamped to the domain bounds, placing diversity outside of the action bounds could be leveraged by the learning agents to bypass the constraint. This is because, placing an increasing amount of diversity in a specific observation, will correspond to the same diversity in the actual actions taken only if the actions are within the domain.

To discourage diversity outside the action bounds, we employ an additional objective:
\begin{equation}
    \mathcal{L}_\mathrm{act}(o) =\max_i \left ( \left \| \mu_i(o) - \mathrm{clip}(\mu_i(o),a_{\mathrm{min}},a_{\mathrm{max}})\right \|_2\right ),
\end{equation}
which penalizes the maximum overflowing action norm over the agents. 

\section{Soft Update $\tau$ Comparison}
\label{app:tau_comparison}

To further evaluate the soft-update mechanism proposed in \autoref{sec:scaling}, we run a comparison across different values of $\tau$ and report the results in \autoref{fig:tau_comparison}. The comparison has been run on the \textit{Multi-Agent Navigation} scenario with $\mathrm{SND}_\mathrm{des} = 1$ for $\tau=0.1,0.01,0.001,0.001$ with 4 different seeds for each value.

It is possible to observe that lower values of $\tau$ result in lower variance of the measured $\mathrm{SND}$ at the price of an overshoot in the mean value, taking longer to converge. Higher values of $\tau$, on the other hand, present a lower error in tracking the mean, at the price of a higher variance and instability around the target value in the earlier stages of training.
The results of this comparison empirically confirm the claim made in \autoref{sec:scaling}.

\begin{figure}[h]
\begin{center}
\centerline{\includegraphics[width=0.4\linewidth]{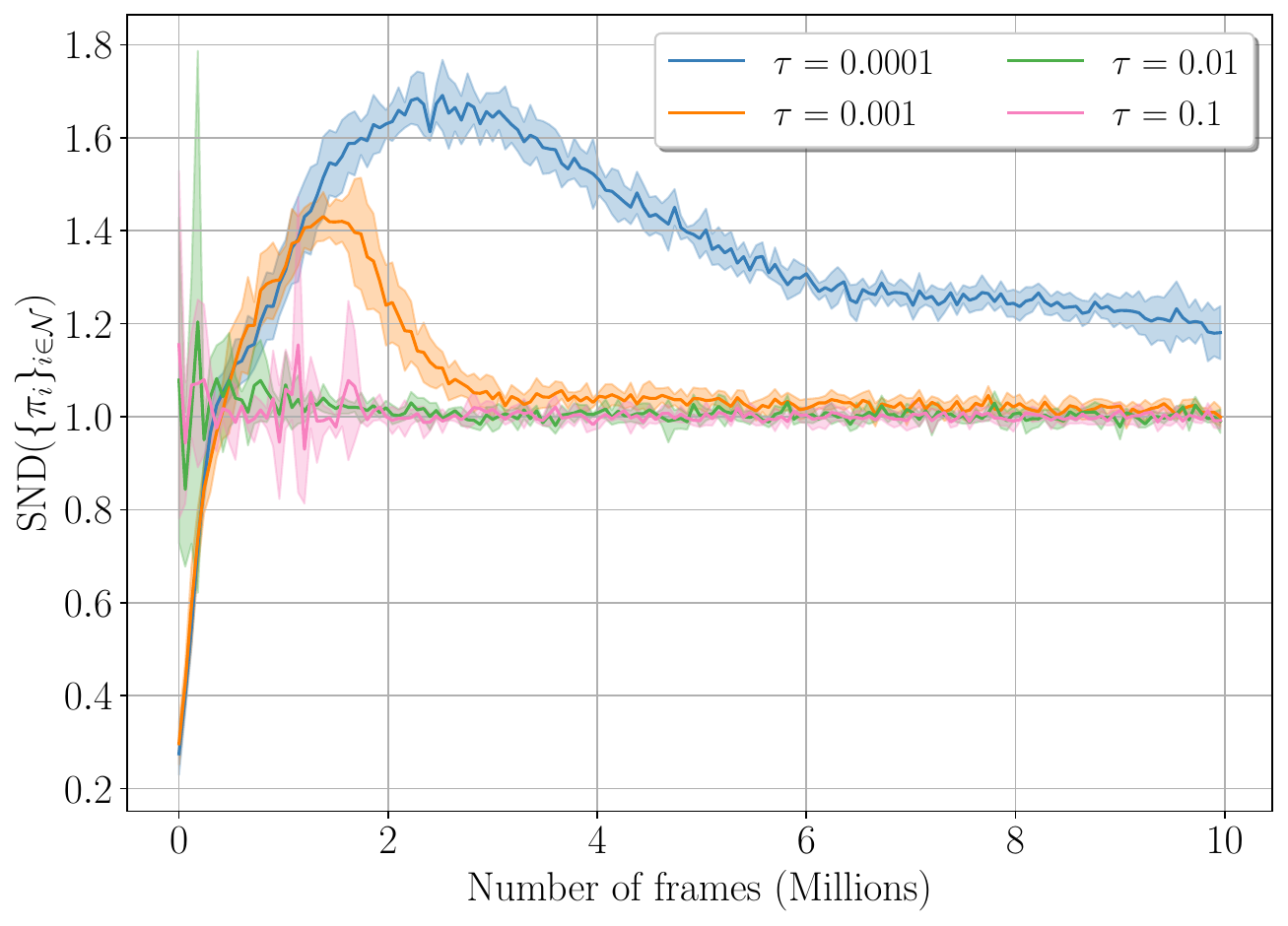}}
\caption{$\tau$ comparison on the \textit{Multi-Agent Navigation} scenario with $\mathrm{SND}_\mathrm{des} = 1$. The comparison shows that lower values of $\tau$ result in lower variance but overshoot the desired diversity, while higher values track the desired diversity with higher variance. Curves report mean and standard deviation for the IPPO algorithm over 4 training seeds.}
\label{fig:tau_comparison}
\end{center}
\end{figure}

\newpage
\section{Further Analyses from Case Study: \textit{Multi-Agent Navigation}}
\label{app:case_study}
This section presents additional experiments and result analyses on the \textit{Multi-Agent Navigation} case study.

\subsection{Setting the Desired Diversity Too High}
\label{app:div_too_high}
\begin{figure*}[th]
\begin{center}
\centerline{\includegraphics[width=\textwidth]{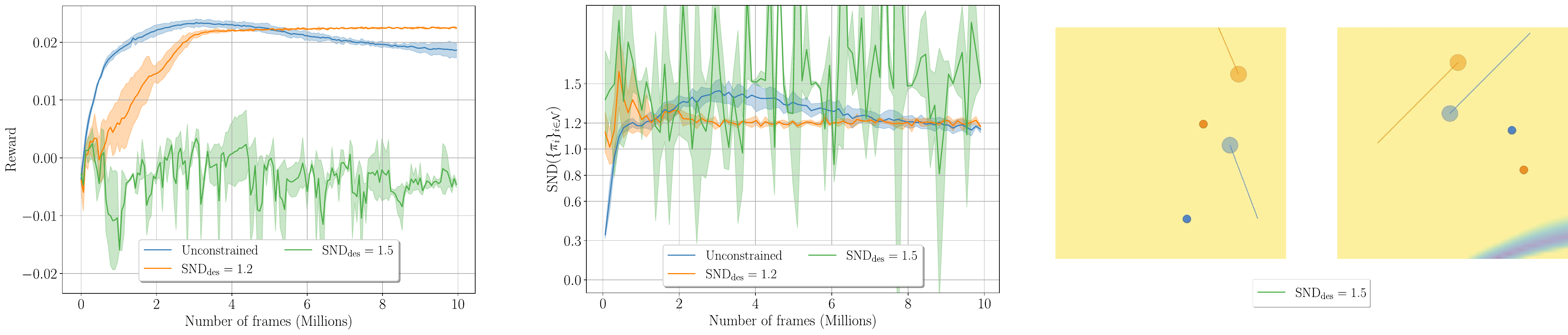}}
\caption{Illustration of the \textit{Multi-Agent Navigation} case study with an example of a high diversity requirement complementing \autoref{fig:nav_case_study}. \textbf{Left:} Mean instantaneous reward for agents trained with different desired diversities. \textbf{Center:} $\mathrm{SND}(\left \{ \pi_{i} \right \}_{i \in \mathcal{N}})$ evaluated  for agents trained with different desired diversities. \textbf{Right:} Renderings of the diversity distribution over the observation space for agents constrained at a high desired diversity. Curves report mean and standard deviation over 4 training seeds.}
\label{fig:nav_case_study_too_much}
\end{center}
\end{figure*}

In \autoref{sec:case_study} we have shown how DiCo can be used to control diversity to different desired values and the impact that this method had on learning in the \textit{Multi-Agent Navigation} task. 
It is however important to discuss what are the implications of setting the desired diversity ``too high''. 
In \autoref{fig:nav_case_study_too_much}, we report the results for training runs with $\mathrm{SND}_\mathrm{des}=1.5$. It is shown how such a diversity value leads to unstable learning in this scenario, with agent policies not able to converge to the desired behavior. This is due to the fact that such a high diversity demand prevents the agents from acting homogeneously even in the observations where this is needed and, thus, never experience and learn from such transitions. Therefore, the agents are not able to learn to distribute diversity over the observation space, resulting in the almost uniform diversity distributions shown in the right of \autoref{fig:nav_case_study_too_much}.

Determining what characterizes an excessive desired diversity depends on the task. To estimate this, it is possible to perform multiple training runs at different desired diversity levels, observing the threshold above which training stability is compromised. To get an estimate of a representative diversity range for a task, it is also possible to train agents with unconstrained diversity and measure their SND value. Lastly, in tasks with bounded action spaces, solving the optimization problem introduced in \autoref{app:max_snd} yields an upper bound to the maximum diversity obtainable in the task.

\subsection{Navigating to the Same Goal}

\begin{figure*}[ht]
\begin{center}
\centerline{\includegraphics[width=\textwidth]{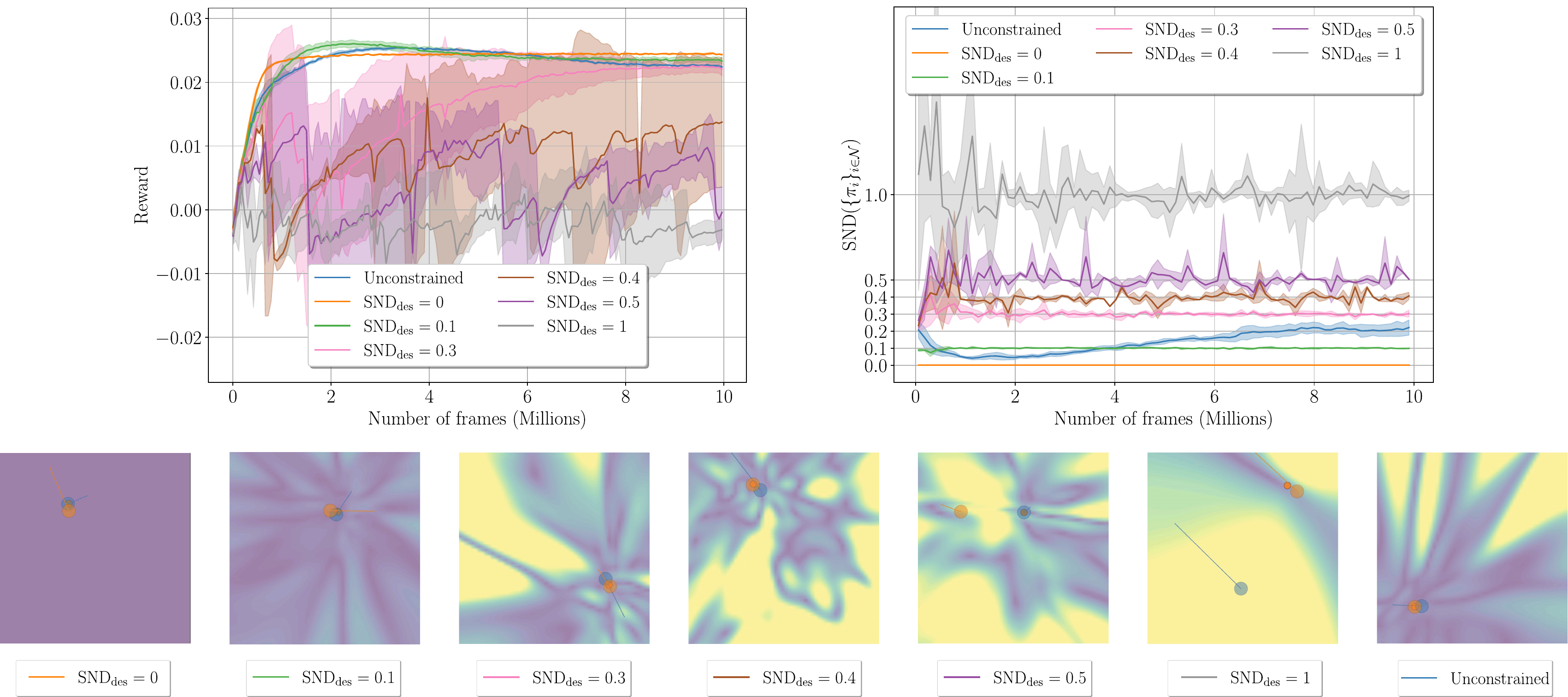}}
\caption{Illustration of the \textit{Multi-Agent Navigation} case study when agents need to navigate to the same goal.  \textbf{Top left:} Mean instantaneous reward for agents trained with different desired diversities. \textbf{Top right:} $\mathrm{SND}(\left \{ \pi_{i} \right \}_{i \in \mathcal{N}})$ evaluated  for agents trained with different desired diversities. \textbf{Bottom:} Renderings of the diversity distribution over the observation space (colormap legend available in \autoref{fig:nav_case_study}) for agents trained with different desired diversities. Agents with a low diversity budget are able to reach the goal with similar policies. As the desired diversity increases, agents need to find different strategies to achieve the same objective, trading off performance for diversity. All strategies present the common characteristic of placing diversity away from the observations near the goal, where it is most important to act homogeneously. Curves report mean and standard deviation for the IPPO algorithm over 4 training seeds.}
\label{fig:nav_case_study_same_goal}
\end{center}
\end{figure*}

The \textit{Multi-Agent Navigation} task, discussed in \autoref{sec:case_study}, represents an example of a task where a certain degree of diversity is needed for successful completion. In this section, we are interested in analyzing a variation of this task where the optimal policy requires homogeneous behavior. In particular, we modify the task such that both agents need to navigate to the same goal (still spawned at random). In this setup, the optimal policy consists in all agents traveling towards the observed goal, requiring a diversity of 0.

As before, we train agents at different desired diversity levels ($\mathrm{SND}_\mathrm{des}=0,0.1,0.3,0.4,0.5,1$) and with unconstrained heterogeneous policies (\autoref{eq:unscaled_policies}). The results are reported in \autoref{fig:nav_case_study_same_goal}. As expected, the measured diversity ($\mathrm{SND}(\{\pi_i\}_{i\in\mathcal{N}})$), shown in the top right, matches the desired values. Looking at the reward curves in the top left, we confirm the hypothesis that heterogeneity is detrimental to performance in this task. In fact, we observe a decrease in reward and training stability as the desired diversity increases, with the best performance achieved by the homogeneous model ($\mathrm{SND}_\mathrm{des} = 0$). It is also noticeable that all models learn to place lower diversity in positions nearing the goal, trying to place the unwanted diversity budget far from this region. Learning this strategy becomes harder as the diversity budget increases, and starts becoming impossible at high diversity values ($\mathrm{SND}_\mathrm{des} = 1$), where only one agent is able to reach the goal and the other is forced to take completely different actions.

We also note that, the unconstrained model, while learning the optimal policies, presents a diversity value varying through training. This is due to the bigger policy search space deriving from the lack of constraints and leads to slower convergence and more instability than low-diversity constrained models ($\mathrm{SND}_\mathrm{des} = 0,0.1$), showcasing another advantage of DiCo, that, thanks to our prior on the heterogeneity required by the task, enforces a low diversity and simplifies learning.

\newpage
\section{\textit{Dispersion} Trajectory Analysis}
\label{app:dispersion}

\begin{figure*}[th]
    \centering
    \subfigure[{$\mathrm{SND}_\mathrm{des}=0$}.]{
        \includegraphics[width=\textwidth]{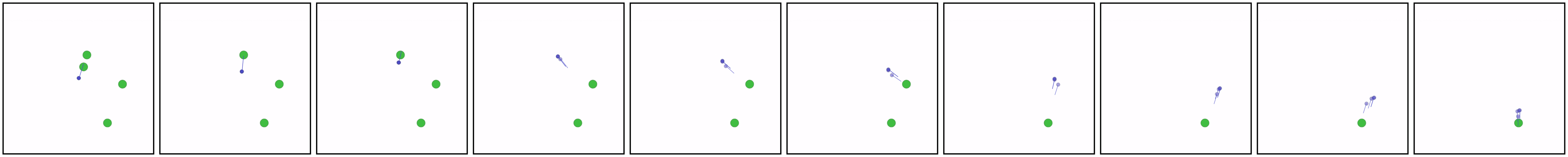}
        \label{fig:dispersion_0_traj}
    }
    \subfigure[{Unconstrained}.]{
        \includegraphics[width=\textwidth]{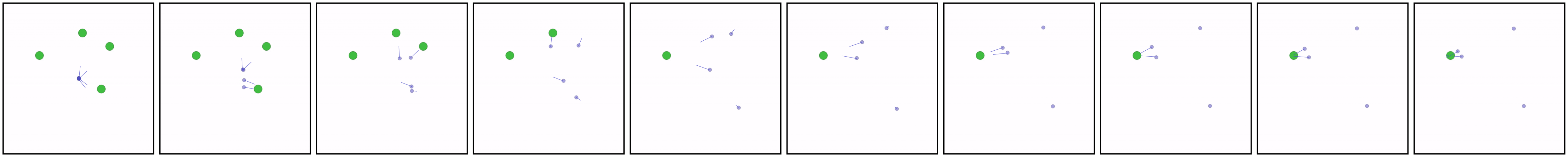}
        \label{fig:dispersion_-1_traj}
    }
     \subfigure[{$\mathrm{SND}_\mathrm{des}=6$}.]{
        \includegraphics[width=\textwidth]{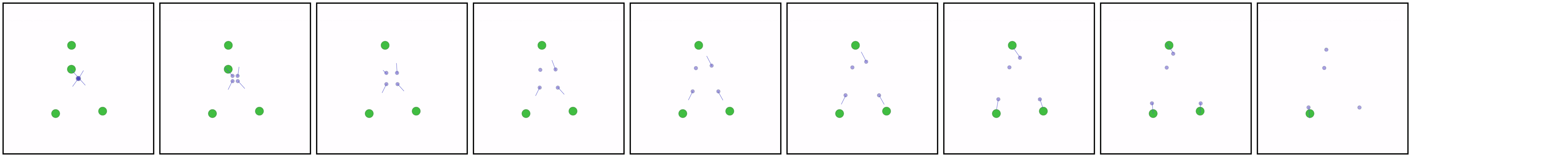}
        \label{fig:dispersion_6_traj}
    }
    \caption{\textit{Dispersion} trajectory renderings from the results in \autoref{fig:dispersion_plots}. Rollouts evolve from left to right. Videos of these trajectories are available on the paper \href{\website}{website}.}
    \label{fig:dispersion_traj}
\end{figure*}

In this section, we complement the \textit{Dispersion} results discussed in \autoref{sec:dispersion} by analyzing the trajectories of the models trained in \autoref{fig:dispersion_plots}. In \autoref{fig:dispersion_traj}, we report the trajectory renderings.

\autoref{fig:dispersion_0_traj} shows the trajectory of homogeneous agents ($\mathrm{SND}_\mathrm{des}=0$). We can observe that homogeneous agents learn to visit one food particle at a time in a group. This is because, given the same observation, they need to take the same action, and, starting in the same position, this process leads them to all take the same actions and visit the same positions. The small variations in positions, visible from the \nth{4} frame onwards, are due to the exploration noise added to incentivize training exploration, that, however, does not help agents learn different policies. 

\autoref{fig:dispersion_-1_traj} shows the trajectory of unconstrained heterogeneous agents (\autoref{eq:unscaled_policies}). These agents are able to tackle different goals, however, we can see that the learned policy is still suboptimal. In fact, while two agents respectively consume the two food particles at the top, the remaining two tackle the same particle in the bottom right. After these three particles have been consumed, two agents (one from the top and one from the bottom) both navigate to the last particle. This behavior obtains a suboptimal reward as it would have taken less time to send each agent to a different particle.

\autoref{fig:dispersion_6_traj} shows the trajectory of heterogeneous agents with constrained diversity $\mathrm{SND}_\mathrm{des}=6$. This diversity constraint forces agent to a higher SND than the one achieved by the unconstrained heterogeneous model. Thanks to this, agents are able to discover sooner the optimal policy. We can see this in the rollout, where agents immediately travel to a different food particle, thus completing the task in the least time needed.

\newpage
\section{\textit{Sampling} Trajectory Analysis}
\label{app:sampling}

\begin{figure*}[th]
    \centering
    \subfigure[{$\mathrm{SND}_\mathrm{des}=0$}.]{
        \includegraphics[width=\textwidth]{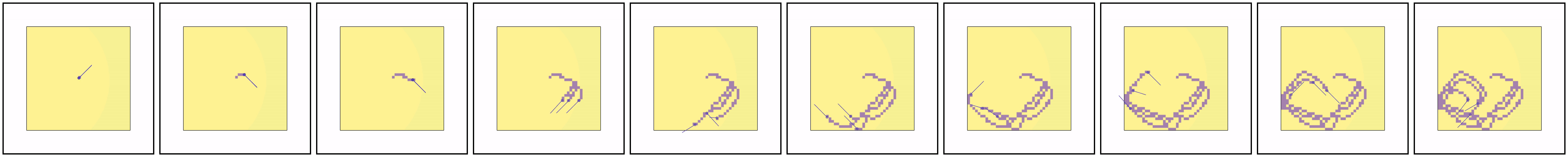}
        \label{fig:sampling_0_traj}
    }
    \subfigure[{Unconstrained}.]{
        \includegraphics[width=\textwidth]{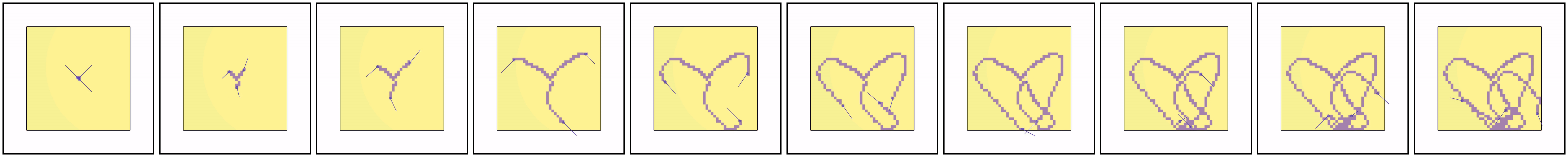}
        \label{fig:sampling_-1_traj}
    }
     \subfigure[{$\mathrm{SND}_\mathrm{des}=5$}.]{
        \includegraphics[width=\textwidth]{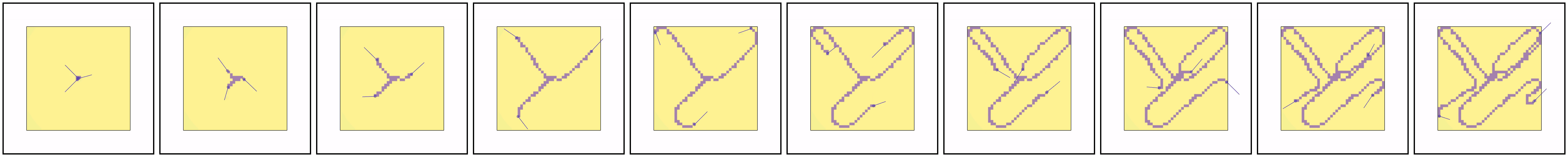}
        \label{fig:sampling_5_traj}
    }
    \caption{\textit{Sampling} trajectory renderings from the results in \autoref{fig:sampling_plots}. Rollouts evolve from left to right. Videos of these trajectories are available on the paper \href{\website}{website}.}
    \label{fig:sampling_traj}
\end{figure*}

In this section, we complement the \textit{Sampling} results discussed in \autoref{sec:sampling} by analyzing the trajectories of the models trained in \autoref{fig:sampling_plots}. In \autoref{fig:sampling_traj}, we report the trajectory renderings.

\autoref{fig:sampling_0_traj} shows the trajectory of homogeneous agents ($\mathrm{SND}_\mathrm{des}=0$).
We can observe that homogeneous agents visit the same area of the sampling space while remaining grouped together. This is because, given the same observation, they need to take the same action, and, starting in the same position, this process leads them to all take the same actions and visit the same positions. The small variations in positions, visible from the \nth{4} frame, are due to the exploration noise added to incentivize training exploration, that, however, does not help agents learn different policies. 

\autoref{fig:sampling_-1_traj} shows the trajectory of unconstrained heterogeneous agents (\autoref{eq:unscaled_policies}). These agents can take different actions for the same observation and, thanks to this, they are able to begin the rollout by moving in different directions. However, from the \nth{5} frame onward, they begin converging towards the same area (in the bottom center of the workspace). By the end of the task, they present a significantly improved coverage than their homogeneous counterparts, but still with some cells that have been visited more than once by different agents, making their policy suboptimal.

\autoref{fig:sampling_5_traj} shows the trajectory of heterogeneous agents with constrained diversity $\mathrm{SND}_\mathrm{des}=5$. This diversity constraint forces agents to a higher SND than the one achieved by the unconstrained heterogeneous model. Under this constraint, the agents achieve a better spread over the sampling space, moving along the workspace diagonals following a zig-zag pattern. They then move back diagonally without ever crossing paths.

An interesting observation can be made by focusing on the first frame of all renderings in \autoref{fig:sampling_traj}. Here we can observe all agents being spawned in the same position in the center of the workspace at the beginning of the task. This is a particularly interesting observation, as all agents are observing the same policy input. This is also one of the observations where they need to act most differently in order to avoid sampling the same cells in the subsequent timestep. As expected, homogeneous agents all take the same action, which is quite suboptimal, leading them all to the same next state. Unconstrained heterogeneous agents improve upon this, with their actions forming two $90$º and one $180$º angles. Constrained heterogeneous agents show the best action spread, almost resembling three $120$º angles.

\newpage
\section{\textit{Tag} Trajectory Analysis}
\label{app:tag}

\begin{figure*}[th]
    \centering
    \subfigure[{$\mathrm{SND}_\mathrm{des}=0$}.]{
        \includegraphics[width=\textwidth]{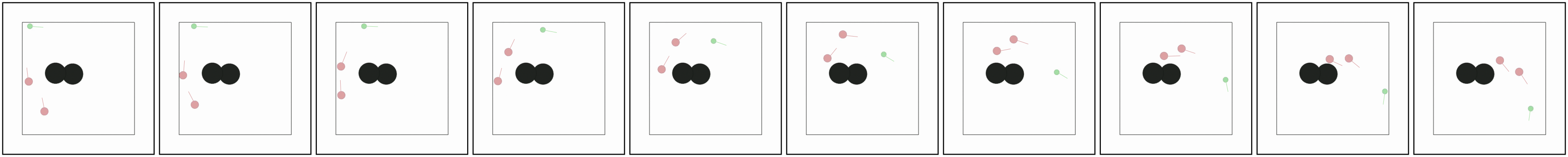}
        \label{fig:tag_0_traj}
    }
    \subfigure[{Unconstrained}.]{
        \includegraphics[width=\textwidth]{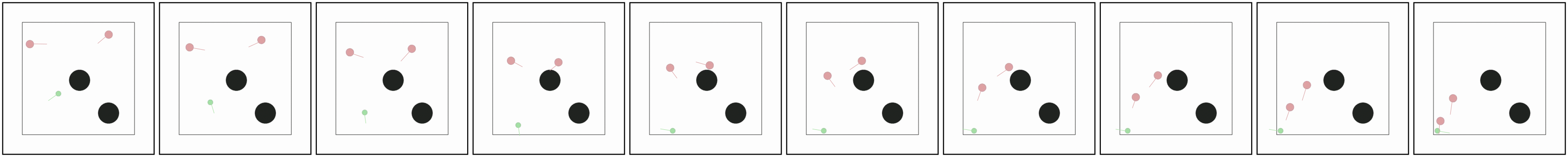}
        \label{fig:tag_-1_traj}
    }
    \subfigure[Ambush emergent strategy ({$\mathrm{SND}_\mathrm{des}=0.6$}).]{
        \includegraphics[width=\textwidth]{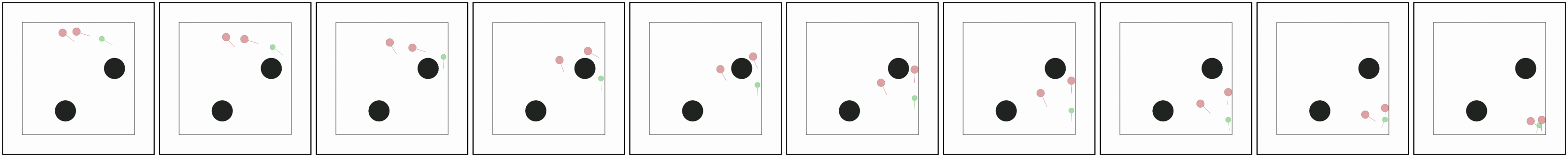}
        \label{fig:tag_ambush_traj}
    }
     \subfigure[Cornering emergent strategy ({$\mathrm{SND}_\mathrm{des}=0.6$}).]{
        \includegraphics[width=\textwidth]{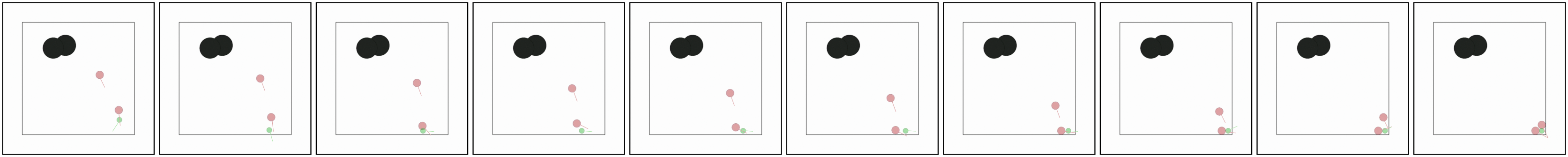}
        \label{fig:tag_cornering_traj}
    }
     \subfigure[Blocking emergent strategy ({$\mathrm{SND}_\mathrm{des}=0.6$}).]{
        \includegraphics[width=\textwidth]{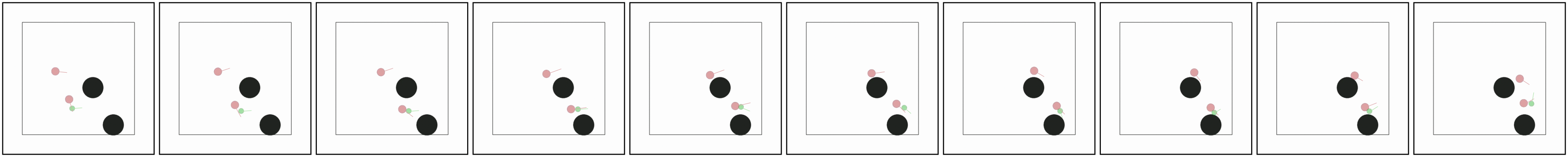}
        \label{fig:tag_blocking_traj}
    }
    \caption{\textit{Tag} trajectory renderings from the results in \autoref{fig:tag_plots}. Rollouts evolve from left to right. Videos of these trajectories are available on the paper \href{\website}{website}.}
    \label{fig:tag_traj}
\end{figure*}

In this section, we complement the \textit{Tag} results discussed in \autoref{sec:tag} by analyzing the trajectories of the models trained in \autoref{fig:tag_plots}. In \autoref{fig:tag_traj}, we report the trajectory renderings.

In \autoref{fig:tag_0_traj} and \autoref{fig:tag_-1_traj} we report the trajectories of homogeneous agents ($\mathrm{SND}_\mathrm{des}=0$) and unconstrained heterogeneous agents (\autoref{eq:unscaled_policies}).
Both models present a similar strategy, consisting in both red agents navigating towards the green one, following the shortest path. This strategy is well known to be suboptimal in ball games like soccer, where inexperienced players tend to group on the ball, resulting in poor spatial coverage.

Continuing with this analogy, we take on a coaching role and constrain agents to a higher diversity ($\mathrm{SND}_\mathrm{des}=0.6$). We observe that this higher diversity constraint induces higher achieved rewards and emergent strategies that leverage agent complementarity. We highlight three examples of such strategies (\autoref{fig:tag_ambush_traj}, \autoref{fig:tag_cornering_traj}, \autoref{fig:tag_blocking_traj}).

In \autoref{fig:tag_ambush_traj}, we observe how the red agents are able to perform an ``ambush'' maneuver, splitting around the obstacle (from the \nth{3} frame). While the right agent performs a close chase, the left agent takes a longer trajectory, which, in the long term, enables it to cut off the green agent. We can see how diversity helped in this scenario: if both agents followed the evader into the passage, the green agent would have had a free escape route on the other side and the chase would have continued in a loop as in \autoref{fig:tag_0_traj}.

In \autoref{fig:tag_cornering_traj}, we observe how the red agents are able to coordinate in a pinching maneuver, cornering the evader in an inescapable state. While one agent lures the evader into a corner (frames 1 to 7), the other agent approaches, progressively closing the escape routes available to the green agent. Eventually, the green agent remains stuck in the corner without any possibility of movement.

Lastly, in \autoref{fig:tag_blocking_traj} we report a strategy where one red agent focuses on blocking the escape routes of the green agent, while the other red agent gathers rewards. In this strategy, more subtle to understand, the red agent in the top performs man-to-man marking (similar to what is done in sports) at a distance. As the green agent moves, this agent tracks its position, blocking access to the top half of the environment by being ready to intercept any movement in that direction.
In doing so, it effectively reduces the maneuvering space available to the green agent, making tagging easier for the other red agent.

\newpage
\section{\textit{Reverse Transport}: Injecting a Prior Through Diversity Constraints}
\label{sec:rev_transport}

\begin{figure}[t]
\begin{center}
\centerline{\includegraphics[width=\linewidth]{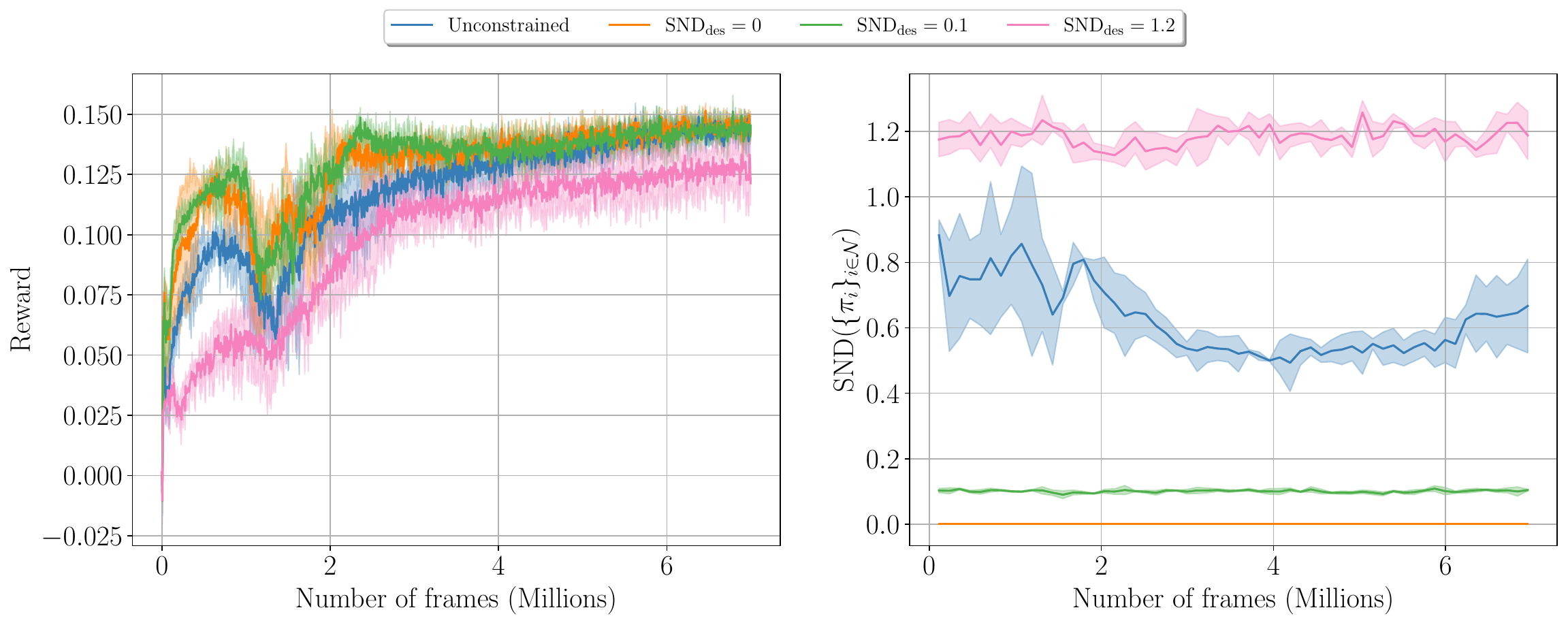}}
\caption{Results from training agents with different constraints on the \textit{Reverse Transport} task. \textbf{Left:} Mean instantaneous reward. \textbf{Right:} Measured diversity $\mathrm{SND}(\left \{ \pi_{i} \right \}_{i \in \mathcal{N}})$. 
Curves report mean and standard deviation for the IDDPG algorithm over 3 training seeds.}
\label{fig:rev_transport_plots}
\end{center}
\end{figure}

In this work, we have mainly shown applications where diversity constraints that force a \textit{higher} diversity than the one achieved by unconstrained policies can be used to obtain higher performance. In this section, we show that, by leveraging user priors about the role of heterogeneity in a task, it is also possible to improve the training process by forcing a \textit{lower} diversity than the one achieved by the unconstrained method. 

We consider the \textit{Reverse Transport} (\autoref{fig:rev_transport}) task. In this task, $n$ agents are placed inside a square red package, spawned at random. The agents are collectively rewarded for pushing the package towards a randomly spawned goal. 
They observe their position, velocity, relative distance to the package, and the package's relative distance to the goal.
Agent actions are 2D forces that determine their motion.
Our hypothesis is that the optimal policy in this task requires little heterogeneity, as it is beneficial for the agents to all perform the same pushing action.

We train $n=4$ agents using IDDPG with different DiCo constraints, as well as unconstrained heterogeneous polices.
In \autoref{fig:rev_transport_plots}, we report the mean instantaneous training reward and the measured diversity.
The results show that the agents with a lower diversity ($\mathrm{SND}_\mathrm{des} = 0,0.1$) achieve the optimal reward faster. Unconstrained heterogeneous policies learn a higher diversity, resulting in slower convergence. Furthermore, DiCo models with a diversity higher than the unconstrained ones ($\mathrm{SND}_\mathrm{des} = 1.2$)  show even more suboptimal curves, confirming our prior on the required homogeneity for the task.  
This shows that the greater sample efficiency of homogeneous agents can be leveraged by setting low DiCo constraints in tasks where a prior on the required heterogeneity is available. While the policies with $\mathrm{SND}_\mathrm{des} = 0.1$ achieve similar rewards to the homogeneous ones, they allow for a small degree of diversity, which might be preferred due to its intrinsic resilience properties~\cite{bettini2023hetgppo}.

\newpage
\section{Analytical Diversity Control}
\label{app:finn}

In this paper, we control diversity by scaling deviations of heterogeneous agents from a homogeneous policy, normalized by an empirically calculated $\widehat{\mathrm{SND}}$. The set of observations $O$ over which this $\widehat{\mathrm{SND}}$ is calculated dictates the distribution of states where the heterogeneity constraint is enforced. In many applications, this is useful---by selecting $O$ as a set of observations from rollouts, we can ensure that all heterogeneity is placed in regions that will actually be visited. However, there exist other applications where one might wish to define an analytical diversity constraint over the whole observation space (note that this can result in a different measured $\mathrm{SND}(\{\pi_i\}_{i\in\mathcal{N}})$ than the given $\mathrm{SND}_\mathrm{des}$, as samples from rollouts do not form a uniform distribution over the observation space).

In this section, we propose a method for applying \textit{analytical} diversity constraints to a system. The core of this approach is an architecture that fundamentally incorporates a diversity constraint, without the need to perform a sum over a set of observations $O$, as in \autoref{eq:snd}. In order to formulate this new method, we must first redefine an integral-based form of $\mathrm{SND}$:

\begin{equation}
    \mathrm{SND}_I(\left \{ \pi_i \right \}_{i \in \mathcal{N}}) = \frac{2}{n(n-1)|\mathcal{O}|}\sum_{i=1}^n\sum_{j=i+1}^n \int_{\mathcal{O}}W_2(\pi_i(o),\pi_j(o)) \, do.
    \label{eq:snd_continuous}
\end{equation}

For convenience, we also recall the definition of $\mathrm{SND}_o$ from \autoref{sec:case_study}, that measures $\mathrm{SND}$ for a given observation $o$:

\begin{equation}
   \mathrm{SND}_{o}(o)= \mathrm{SND}_o(\left \{ \pi_{h,i} \right \}_{i \in \mathcal{N}},o) =
  \frac{2}{n(n-1)}\sum_{i=1}^n\sum_{j=i+1}^n W_2(\pi_{h,i}(o),\pi_{h,j}(o)).
\label{eq:sndhat_continuous}
\end{equation}

Now, we must find a way to compute $\int_{\mathcal{O}}W_2(\pi_i(o),\pi_j(o)) \, do$, the integral of an expression containing learned functions $\{ \pi_i \}_{i \in \mathcal{N}}$ over the entire observation space. To do this, we use a Fixed Integral Neural Network (FINN) \cite{finn}, a method that enables the analytical evaluation of a learned multivariate function $f(x)$ and its integral $F(x) |_A = \int_A f(x) \, dx$. In this application, $f$ represents the amount of system heterogeneity allocated to a given observation:

\begin{equation}
    f(o) = \frac{2}{n(n-1)}\sum_{i=1}^n\sum_{j=i+1}^n W_2(\pi_i(o),\pi_j(o)).
    \label{eq:finn_f}
\end{equation}

Integrating $f(o)$ over the observation space $O$, $\frac{F(o)|_\mathcal{O}}{|\mathcal{O}|}$ is equivalent to $\mathrm{SND}_I$ from \autoref{eq:snd_continuous}. Rescaling to achieve a particular diversity is handled by FINN, which can internally place constraints over $F$. We supply an equivalence constraint with the value $|\mathcal{O}| \, \mathrm{SND}_\mathrm{des}$, ensuring $\frac{F(o)|_\mathcal{O}}{|\mathcal{O}|}=\mathrm{SND}_I(\{\pi_i\}_{i\in\mathcal{N}})=\mathrm{SND}_\mathrm{des}$. Using FINN's positivity constraint, we further constrain $f$ to be strictly non-negative, as it is impossible for an $\mathrm{SND}_I$ to be less than $0$.

Given these preliminaries, we can reformulate the policies in terms of $f$:

\begin{equation}
    \pi_i(o) = \pi_h(o) + \frac{f(o)}{\mathrm{SND}_o(o)}\pi_{h,i}(o).
    \label{eq:finn_pi}
\end{equation}

\begin{theorem}
    Using the analytical formulation of DiCo in \autoref{eq:finn_pi} and the definition of $\mathrm{SND}_I$ in \autoref{eq:snd_continuous}, the property $\mathrm{SND}_I(\{\pi_i\}_{i\in\mathcal{N}}) = \mathrm{SND}_\mathrm{des}$ is satisfied.
\end{theorem}

\begin{proof}
    We wish to show that the $\mathrm{SND}_I$ (\autoref{eq:snd_continuous}) computed with the analytical formulation of DiCo (\autoref{eq:finn_pi}) equals $\mathrm{SND}_\mathrm{des}$. 

    \begin{equation}
    \frac{2}{n(n-1)|\mathcal{O}|}\sum_{i=1}^n\sum_{j=i+1}^n \int_{\mathcal{O}}W_2 \left (\pi_h(o) + \frac{f(o)}{\mathrm{SND}_o(o)}\pi_{h,i}(o),\pi_h(o) + \frac{f(o)}{\mathrm{SND}_o(o)}\pi_{h,j}(o) \right ) \, do = \mathrm{SND}_\mathrm{des}.
    \label{eq:snd_cont_proof_1}
    \end{equation}

    First, we rewrite \autoref{eq:snd_cont_proof_1} with the definition of $W_2$ from \autoref{eq:wasserstein}. We focus on the case of deterministic policies. Proofs for the other cases can be obtained by following the steps in \autoref{sec:proofs}.

    \begin{equation}
    \frac{2}{n(n-1)|O|}\sum_{i=1}^n\sum_{j=i+1}^n \int_{\mathcal{O}}
    \sqrt{\left \|\left ( \mu_h(o) + \frac{f(o)}{\mathrm{SND}_o(o)}\mu_{h,i}(o)\right )-\left ( \mu_h(o) + \frac{f(o)}{\mathrm{SND}_o(o)}\mu_{h,j}(o)\right ) \right \|^2_2} \, do \\ = \mathrm{SND}_\mathrm{des}.
    \end{equation}

    The $\mu_h(o)$ terms cancel, allowing us to pull $f(o)$ and $\mathrm{SND}_o(o)$ out of the root. We also move the summations and leading constant term into the integral:

    \begin{equation}
    \frac{1}{|\mathcal{O}|}\int_{\mathcal{O}} f(o)
    \frac{\frac{2}{n(n-1)} \sum_{i=1}^n\sum_{j=i+1}^n \left \|\mu_{h,i}(o)-\mu_{h,j}(o) \right \|_2}{\mathrm{SND}_o(o)} \, do = \mathrm{SND}_\mathrm{des}.
    \end{equation}

    Note that the numerator of the fractional term is equivalent to $\mathrm{SND}_o(o)$, so the entire fraction cancels. We replace $\int_\mathcal{O} f(o) do$ with $F(o) \big\rvert_\mathcal{O}$, and multiply both sides by $|\mathcal{O}|$:

    \begin{equation}
    F(o) \big\rvert_\mathcal{O} = |\mathcal{O}| \, \mathrm{SND}_\mathrm{des}.
    \end{equation}

    We are left with the constraint that we imposed, which is guaranteed by FINN. Thus, the system is analytically constrained to the desired diversity.
\end{proof}

\begin{figure*}[h]
\begin{center}
\centerline{\includegraphics[width=\textwidth]{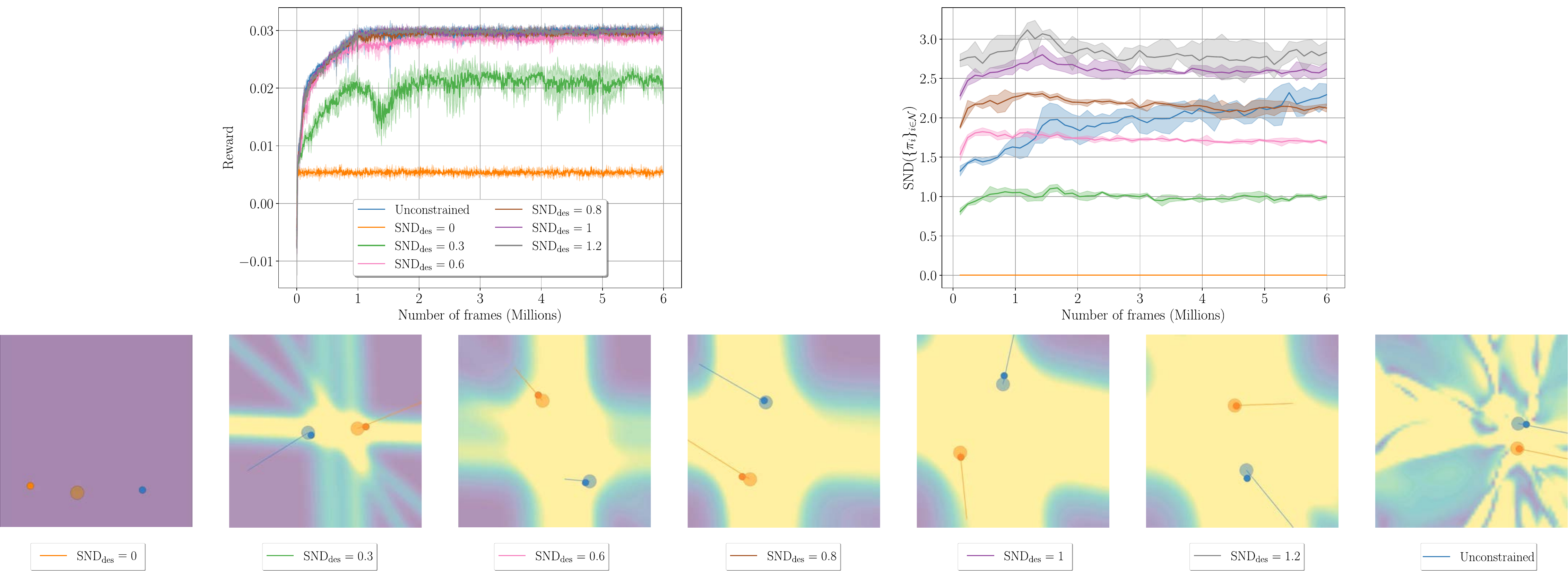}}
\caption{Illustration of the \textit{Multi-Agent Navigation} case study with analytical diversity control.  \textbf{Top left:} Mean instantaneous reward for agents trained with different desired diversities. \textbf{Top right:} $\mathrm{SND}(\left \{ \pi_{i} \right \}_{i \in \mathcal{N}})$ evaluated  for agents trained with different desired diversities. \textbf{Bottom:} Renderings of the diversity distribution over the observation space (colormap legend available in \autoref{fig:nav_case_study}) for agents trained with different desired diversities. Curves report mean and standard deviation for the MADDPG algorithm over 3 training seeds.}
\label{fig:nav_finn}
\end{center}
\end{figure*}

As a preliminary experiment, we train analytical DiCo on the \textit{Multi-Agent Navigation} task with heterogeneous goals. As shown in \autoref{fig:nav_finn}, the reward scales with increasing diversity, just as it does in the same task with standard DiCo. The visualizations show that the placement of the diversity is also similar to the standard approach, whereby the space between the goals is prioritized. Note that the value of the measured $\mathrm{SND}$ does not exactly match the given $\mathrm{SND}_\mathrm{des}$, as the $\mathrm{SND}_I$ constraint is guaranteed over the entire observation space, while the calculated $\mathrm{SND}$ is only computed over observations seen in rollouts. However, the \textit{relative} change in $\mathrm{SND}$ is consistent, with higher $\mathrm{SND}_\mathrm{des}$ leading to higher measured SND.

This experiment empirically demonstrates the ability of analytical DiCo to apply a neural diversity constraint across the entire observation space. This theoretical guarantee can be useful in offline RL applications, where it is imperative to generalize to states that were not visited during training. It also provides a new way of viewing SND: not as a metric that is computed over a set of observations, but rather as a fundamental property of the set of policies themselves.

\end{document}